\documentclass{sig-alternate}
\usepackage{graphicx}
\usepackage{balance}  

\usepackage{mathtools,amssymb}

\usepackage[linesnumbered,ruled,procnumbered]{algorithm2e}
\usepackage{times}
\usepackage{url}
\usepackage{color}
\usepackage{caption}
\usepackage{bbding}
\usepackage{moresize}
\usepackage{multirow}
\usepackage{colortbl}
\usepackage{makecell}
\usepackage{enumitem}
\usepackage{flushend} 
\usepackage{moresize}
\usepackage{subfigure}

\usepackage[colorlinks=true,citecolor=blue,linkcolor=blue]{hyperref}

\usepackage{graphicx}
\usepackage{subfigure}

\newcommand{\ownfull}{{\small \sf Ownership Relaying (OR)}}
\newcommand{\own}{{\small \sf OR}}

\newcommand{\tph}{{L-Store}}
\newcommand{\tpfull}{{\small \sf Lineage-based Data Store (L-Store)}}
\newcommand{\tp}{{\small \sf L-Store}}
\newcommand{\tps}{{\ssmall L-Store}}

\newcommand{\iph}{{\small \sf In-place Update + History}}
\newcommand{\dbt}{{\small \sf Delta + Blocking Mer\-ge}}

\newcommand{\cm}{{\CheckmarkBold}}
\newcommand{\xm}{{ -- }}

\usepackage{lipsum}
\usepackage{ntheorem}
\newtheorem*{steps}{}

\newtheorem{theorem}{Theorem}
\newtheorem{lemma}{Lemma}

\newif\ifcomments 

\newcommand{\comment}[1]{}
\newcommand{\eat}[1]{} 
\newcommand{\todo}[1]{}
\ifcomments
  \renewcommand{\comment}[1]{{\em [#1]}}
  \renewcommand{\todo}[1]{{\em [TODO: #1]}}
\fi

\makeatletter
\def\@copyrightspace{\relax}
\makeatother

\title{L-Store: A Real-time OLTP and OLAP System}

\author{
Mohammad Sadoghi$^{\dag}$, Souvik Bhattacherjee$^{\ddag}$\titlenote{Work was performed as part of a summer internship at IBM T.J. Watson Research Center under Mohammad Sadoghi's mentorship.}, Bishwaranjan Bhattacharjee$^{\#}$, Mustafa Canim$^{\#}$ \\
\affaddr{$^{\dag}$Department of Computer Science, Purdue University}\\
\affaddr{$^{\ddag}$University of Maryland, College Park}\\
\affaddr{$^{\#}$IBM T.J. Watson Research Center}\\
}

\begin{document}

\maketitle

\begin{abstract}
Arguably data is the new natural resource in the enterprise world with an 
unprecedented degree of proliferation. But to derive real-time actionable 
insights from the data, it is important to bridge the gap between 
managing the data that is being updated at a high velocity (i.e., OLTP) and 
analyzing a large volume of data (i.e., OLAP). However, there has been a divide 
where specialized solutions were often deployed to support either OLTP or OLAP 
workloads but not both; thus, limiting the analysis to stale and possibly 
irrelevant data. In this paper, we present \tpfull\ that combines the real-time 
processing of transactional and analytical workloads within a single unified 
engine by introducing a novel lineage-based storage architecture. By exploiting 
the lineage, we develop a contention-free and lazy staging of columnar data from 
a write-optimized form (suitable for OLTP) into a read-optimized form (suitable for 
OL\-AP) in a transactionally consistent approach that also supports querying and 
retaining the current and historic data. Our working prototype of \tp\ demonstrates 
its superiority compared to state-of-the-art approaches under a 
comprehensive experimental evaluation.
\end{abstract}
\section{Introduction}
We are now witnessing an architectural shift and divide in database 
community. The first school of thought emerged from an academic conjecture 
that ``\textit{one size does not fit all}'' \cite{Stonebraker:2005} (i.e., \textit{advocating specialized 
solutions}), which has lead to manifolds of innovations over the last decade in 
creating specialized and subspecialized database engines geared toward various 
niche workloads and application scenarios (e.g., \cite{Stonebraker:2005,DBLP:conf/cidr/BonczZN05,DBLP:journals/pvldb/KallmanKNPRZJMSZHA08,erling2012virtuoso,Raman:2013,Diaconu:2013,ibmsoliddb,Stonebraker_thevoltdb,Rabl:2012:SBD:2367502.2367512}). 
This school has successfully influenced major database vendors such as Microsoft to focus on building 
new specialized engines offered as loosely integrated engines (e.g., Hekaton in-memory engine~\cite{Diaconu:2013} 
and Apollo column store engine~\cite{Larson:2011:SSC:1989323.1989448}) within a single umbrella of database portfolio 
(notably, recent efforts are now focused on a tighter real-time integration of 
Hekaton and Apollo engines~\cite{Larson:2015}). It has also influenced Oracle to partially 
accept the basic premise that ``one size does not fit all'' as far as data representation 
is concerned and has led Oracle to develop a dual-format technique~\cite{LahiriOracle} that 
maintains two tightly integrated representation of data (i.e., two copies of the data) 
in a transactionally consistent manner. 

However, the second school of thought, supported by both aca\-dem\-ia (e.g.,~\cite{Kemper:2011,Cao:2011,chen2010providing,DBLP:conf/sigmod/AlagiannisIA14,hyperMVCC,DBLP:conf/sigmod/LangMFB0K16})
and industry (e.g., SAP \cite{DBLP:journals/debu/FarberMLGMRD12}), 
rejects the aforementioned fundamental premise and advocates a generalized solution. Proponents 
of this idea, rightly in our view, make the following arguments. First, there is a tremendous 
cost in building and maintaining multiple engines from both the perspective of database 
vendors and users of the systems (e.g., application development and deployment costs). 
Second, there is a compelling case to support real-time decision making on the 
latest version of the data~\cite{Plattner14} (likewise supported by~\cite{LahiriOracle,Larson:2015}), 
which may not be feasible across loosely integrated engines that are connected through 
the extract-transform load (ETL) process. Closing this gap may be possible, but its 
elimination may not be feasible without solving the original problem of unifying OLTP and OLAP
capabilities or without being forced to rely on ad-hoc approaches to bridge the gap in 
hindsight. We argue that the separation of OLTP and OLAP capabilities is a step backward 
that defers solving the actual challenge of real-time analytics. Third, combining real-time 
OLTP and OLAP functionalities remains as an important basic research question, which demands deeper 
investigation even if it is purely from the theoretical standpoint. 

In this dilemma, we support the latter school of thought (i.e., \textit{advocating 
a generalized solution}) with the goal of undertaking an 
important step to study the entire landscape of single engine architectures 
and to support both transactional and analytical workloads holistically 
(i.e., ``\textit{one size fits all}''). In this paper, we present \tpfull\ with 
a novel lineage-based storage architecture to address the conflicts between row- and 
column-major representation by developing a contention-free and lazy staging of 
columnar data from write optimized into read optimized form in a transactionally 
consistent manner without the need to replicate data, to maintain multiple representation 
of data, or to develop multiple loosely integrated engines that sacrifices real-time 
capabilities.

\begin{table*}
\centering
\begin{ssmall}
\begin{tabular}{| l || c | c | c | c | c | c | c |}  \hline
					& \tps\	& HANA~\cite{Plattner14}&	ES2~\cite{Cao:2011,chen2010providing}	& HyPer~\cite{Kemper:2011}	& Oracle Dual-format~\cite{LahiriOracle} & Microsoft SQL Server~\cite{Larson:2015} & HBase+Hadoop \\ [0.5ex] \hline\hline
%

Single Product		& \cm\	& \cm\		& \cm\ 		& \cm\		& \cm\ 				& \cm\					&	\xm\   \\ \hline
Single Engine		& \cm\	& \cm\		& \cm\ 		& \cm\		& \cm\				& \xm\					&	\xm\    \\ \hline
Single Instance		& \cm\	& \cm\		&  \cm\		& \cm\		& \cm\				& \cm\					&	\xm\   \\ \hline
Single Copy			& \cm\	& \cm\		&  \cm\		& \makecell{data page replication \\ (OS forking)}		& \xm\				& \xm\					&	\xm\  \\	\hline	
Single Representation& \cm\	& main + delta			&  \cm\		& \cm\ \comment{\makecell{configurable \\ (row + columnar)}}& \xm\				& \xm\					&	\xm\  \\	\hline	
OLTP-optimized 		& \cm\	& \cm\		&  \makecell{limited to \\ get/put operations}	& \makecell{limited to \\ partitionable workloads \\ (i.e., serial execution)}		& \cm\				& \cm\					&	\cm\  \\  \hline
OLAP-optimized		& \cm\	& \cm\		&  \makecell{inconsistent snapshot \\ is possible}		& \cm\		& \cm\				& \cm\					&	\cm\  \\  \hline
Unified OLTP+OLAP	& \cm\	& \cm\		&  \cm\ 	& \cm\ 	& \cm\ 				& \cm\					&	\xm\  \\  \hline
\end{tabular}
\end{ssmall}
\caption{Architectural characterization of selected database engines.}
\label{tab:comparison}
\end{table*}

To further disambiguate our notion of ``\textit{one size fits all}'', in this paper, 
we restrict our focus to real-time relational OLTP and OLAP capabilities. We define a 
set of architectural characteristics for distinguishing the differences between 
existing techniques. First, there could be a single product consisting of multiple 
loosely integrated engines that can be deployed and configured to support either 
OLTP or OLAP. Second, there could be a single engine as opposed to having multiple 
specialized engines packaged in a single product. Third, even if we have a single 
engine, then we could have multiple instances running over a single engine, 
where one instance is dedicated and configured for OLTP workloads while another 
instance is optimized for OLAP workloads, in which the different instances are 
assumed to be connected using an ETL process. Finally, even when using the same 
engine running a single instance, there could be multiple copies or representations  
(e.g., row vs. columnar layout) of the data, where one copy (or representation) 
of the data is read optimized while the second copy (or representation) is write 
optimized. The architectural comparison of various existing techniques based on 
our rigorous definition of ``one size fits all'' is outlined in Table~\ref{tab:comparison}.\footnote{However, 
it is crucial to note that the presented comparison is solely focused on the 
overall architectural choices, and it does not make any claims about the relative 
system performance and/or functionalities. For example, if HANA contains more check 
marks than Microsoft SQL Server, it does not imply that HANA is a better product, 
instead it simply assesses HANA architecturally with respect to our definition of 
``one size fits all''.}

In short, we develop \tp, 
an important first step towards supporting real-time OLTP and OLAP 
processing that faithfully satisfies our definition of \textit{generalized solution}, 
and, in particular, we make the following contributions:
 
\begin{itemize}
\item Introducing a lineage-based storage architecture that enables 
a contention-free update mechanism over a native multi-version, columnar 
storage model in order to lazily and independently stage stable data from a 
write-optimized columnar layout (i.e., OLTP) into a read-optimized 
columnar layout (i.e., OLAP)

\item Achieving (at most) 2-hop away access to the latest version of any 
record (preventing read performance deterioration for point queries)

\item Contention-free merging of only stable data, namely, me\-rging of the 
read-only base data with recently committed updates (both in columnar representation) 
without the need to block ongoing or new transactions by relying on the lineage 

\item Contention-free page de-allocation (upon the completion of the merge process) 
using an epoch-based approach without the need to drain the ongoing transactions

\item A first of its kind comprehensive evaluation to study the leading architectural 
storage design for concurrently supporting short update transactions and analytical queries 
(e.g., an in-place update with a history table architecture and the commonly 
employed main and delta stores architecture)

\end{itemize}

\textbf{Motivating Real-time OLTP and OLAP}
Before describing our proposed approach in-depth, we briefly present two 
important scenarios that benefit greatly from a real-time OLTP and OLAP 
solution.

Consider the mobile e-commerce market, in which the revenue for the 
location-based mobile advertising alone is expected to reach \$18 
billions by 2019~\cite{loc-ad}. A potential buyer with a mobile device may 
roam around physically while shopping. In the meantime, the sh\-opper's mobile 
device generates location information. Alternatively, as the sh\-opper browses the 
web, again the location information is either exchanged explicitly or 
detected automatically based on the shopper's IP address or by its connection 
to the nearby WiFi routers. Now the task of any real-time targeted advertising 
auction is to determine and present a set of relevant ads to the shopper by 
running analytics over the location information, shopping patterns, past 
purchases, and browsing history of the shopper. Furthermore, if these advertisements 
result in a purchase, then the resulting transactions need to become available 
immediately to subsequent analytics in order to improve the effectiveness of 
future advertisements. Moreover, the actual ad bidding in the auction also 
requires a transactional semantics support in real-time. Finally, all these 
steps must be completed typically within 150 milliseconds~\cite{loc-ad-economist}. 
Therefore, we argue that in order to sustain a high velocity transactional data 
(e.g., Google AdWords served almost 30 billion ads per day in 2012~\cite{google-ad}) 
while executing complex analytics on the latest and historic (transactional) data, 
there is a compelling need to develop a solution that exhibits a true real-time 
OLTP and OLAP capabilities.

Another prominent scenario is fraud detection especially at the time when the 
cost of cybercrime continues to increase at a staggering rate and has already 
surpassed \$400 billion dollars annually~\cite{mcafee}. For instance, a credit 
card company will need to approve a transaction in a small time window (i.e., 
subsecond ranges). During this short time span, it is forced to determine if a  
transaction is fraudulent or not. Thus, there is a crucial need to run complex 
analytics in real-time as part of the transaction that is being processed. Without 
such a proactive fraud detection capability, fraudulent transactions may remain 
undetectable, which may result in irreversible financial losses as clearly been 
witnessed when billions of dollars are being lost due to fraud activities every 
year~\cite{mcafee}. Furthermore, there are indirect financial losses involving 
stakeholders such as credit card companies and merchants. The indirect losses 
attributed to decline of legitimate transactions that disrupts merchant's 
daily operation, lost payment volume as consumers opt for alternative 
payment types that are perceived to be safer, and lost customers due to card 
cancellation and reissue~\cite{credit}.

\begin{figure*}[t!]
	\centering
	\begin{minipage}[t]{0.5\linewidth}		
        \centering
        \includegraphics[scale=0.35]{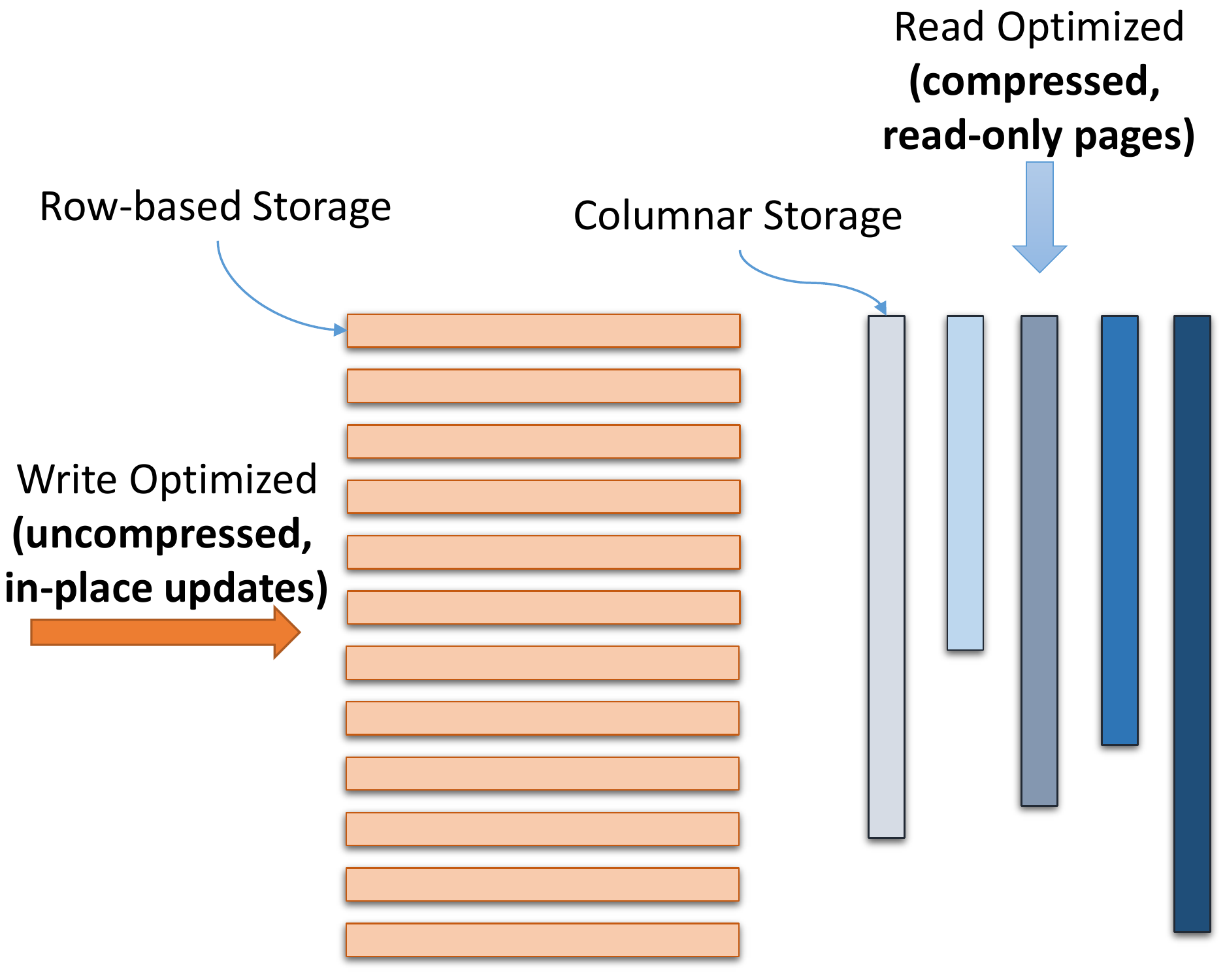}
        \caption{Overview of storage layout conflict.}
        \label{fig:conflict}
	\end{minipage}\hfill
	\begin{minipage}[t]{0.5\linewidth}
        \centering
        \includegraphics[scale=0.37]{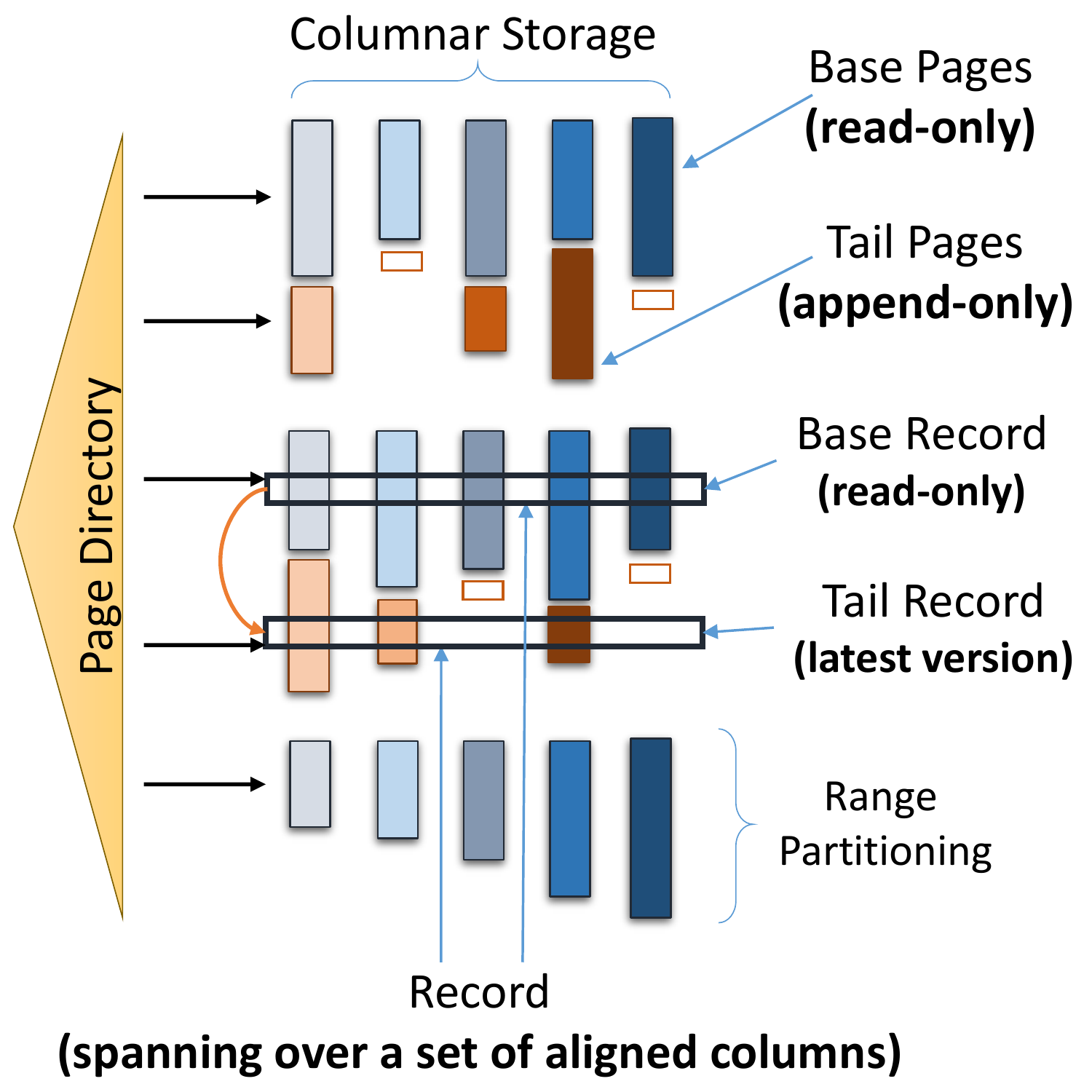}
        \caption{Overview of the lineage-based storage architecture.}
        \label{fig:overview}
	\end{minipage}\hfill
\end{figure*}

\section{Unified Architecture}
The divide in the database community is partly attributed to the 
storage conflict pertaining to the representation of transactional and analytical 
data. In particular, transactional data requires write\--opt\-imiz\-ed storage, 
namely the row-based layout, in which all col\-umns are co-located (and 
preferably uncompressed for in-place updates). This layout improves point update 
mechanisms, since accessing all columns of a record can be achieved by 
a single I/O (or few cache misses for memory-resident data). In contrast, 
to optimize the analytical workloads (i.e., reading many records), it 
is important to have read-optimized storage,  i.e., columnar layout in 
highly compressed form. The intuition behind having columnar layout is 
due to the observation that most analytical que\-ries tend to access only 
a small subset of all columns~\cite{pax}. Thus, by storing data column-wise, 
we can avoid reading irrelevant columns (i.e., reducing the raw amount of 
data read) and avoid polluting processor's cache with irrelevant data, 
which substantially improve both disk and memory bandwidth, respectively. 
Furthermore, storing data in columnar form improves the data homogeneity 
within each page, which results in an overall better compression ratio.
This storage conflict is depicted in Figure~\ref{fig:conflict}.

\subsection{\tph\ Storage Overview}
To address the dilemma between write- and read-optimized layouts, we develop \tp. 
As demonstrated in Figure~\ref{fig:overview}, the high-level architecture of 
\tp\ is based on native columnar layout (i.e., data across columns are aligned to 
allow implicit re-construction), where records are (virtually) partitioned 
into disjoint ranges (also referred to as update range). Records within each range 
span a set of read-only, compressed pages, which we refer to them as the \textit{base pages}. 
More importantly, for every range of records, and for each updated column 
within the range, we maintain a set of app\-end-only pages to store the latest 
updates, which we refer to them as the \textit{tail pages}. Anytime a record 
is updated in base pages, a new record is appended to its corresponding tail pages, 
where there are explicit values only for the updated columns (non-updated 
columns are preassigned a special null value when a page is first allocated). 
We refer to the records in base pages as the \textit{base records} and the records 
in tail pages as the \textit{tail records}. Each record (whether falls in base 
or tail pages) spans over a set of aligned columns (i.e., no join is necessary 
to pull together all columns of the same record).\footnote{Fundamentally, there 
is no difference between base vs. tail record, the distinction is made only 
to ease the exposition.} 

A unique feature of our lineage-based architecture is that tail pages are strictly append-only and 
follow a write-once policy. In other words, once a value is written to tail pages, 
it will not be over-written even if the writing transaction aborts. The append-only design together 
with retaining all versions of the record substantially simplifies low-level synchronization and recovery 
protocol (as described in Section~\ref{sec_rec}) and enables efficient realization of multi-version 
concurrency control. Another important property of our lineage-based storage is that all data are 
represented in a common holistic form; there are no ad-hoc corner cases. Records in both base and 
tail pages are assigned record-identifiers (RIDs) from the same key space. Therefore, both base and 
tail pages are referenced through the database page directory using RIDs and persisted 
identically. Therefore, at the lower-level of the database stack, there is absolutely no 
difference between base vs. tail pages or base vs. tail records; they are presented 
and maintained identically.  

To speed query processing, there is also an explicit linkage (forward and backward 
pointers) among records. From a base record, there is a forward pointer to the latest version of 
the record in tail pages. The different versions of the same records in  
tail pages are chained together to enable fast access to an earlier version of the 
record. The linkage is established by introducing a table-embedded indirection 
column that stores forward pointers (i.e., RIDs) for base records and backward pointers 
for tail records (i.e., RIDs).

The final aspect of our lineage-based architecture is a periodic, con\-tent\-ion-free merging of a set of base 
pages with its corresponding tail pages. This is performed to consolidate base pages 
with the recent updates and to bring base pages forward in time (i.e., creating a set 
of merged pages). Each merged page independently maintains its lineage information, i.e., 
keeping track of all tail records that are consolidated onto the page thus far. By maintaining explicit 
in-page lineage information, the current state of each page can be determined independently, 
and the base page can be brought up to any desired snapshot. Tail pages that are already merged 
and fall outside the snapshot boundaries of all active queries are called historic tail-pages. 
These pages are re-organized, so that different versions of a record are stored 
contiguously inlined. Delta-com\-pression is applied across different versions of 
tail records, and tail records are ordered based on the RIDs of their corresponding 
base records. Below, we describe the unique design and algorithmic features of 
\tp\ that enables efficient transactional processing without performance deterioration 
of analytical processing; thereby, achieving a real-time OLTP and OLAP.

\begin{figure*}
\begin{center}
\includegraphics[scale=0.435]{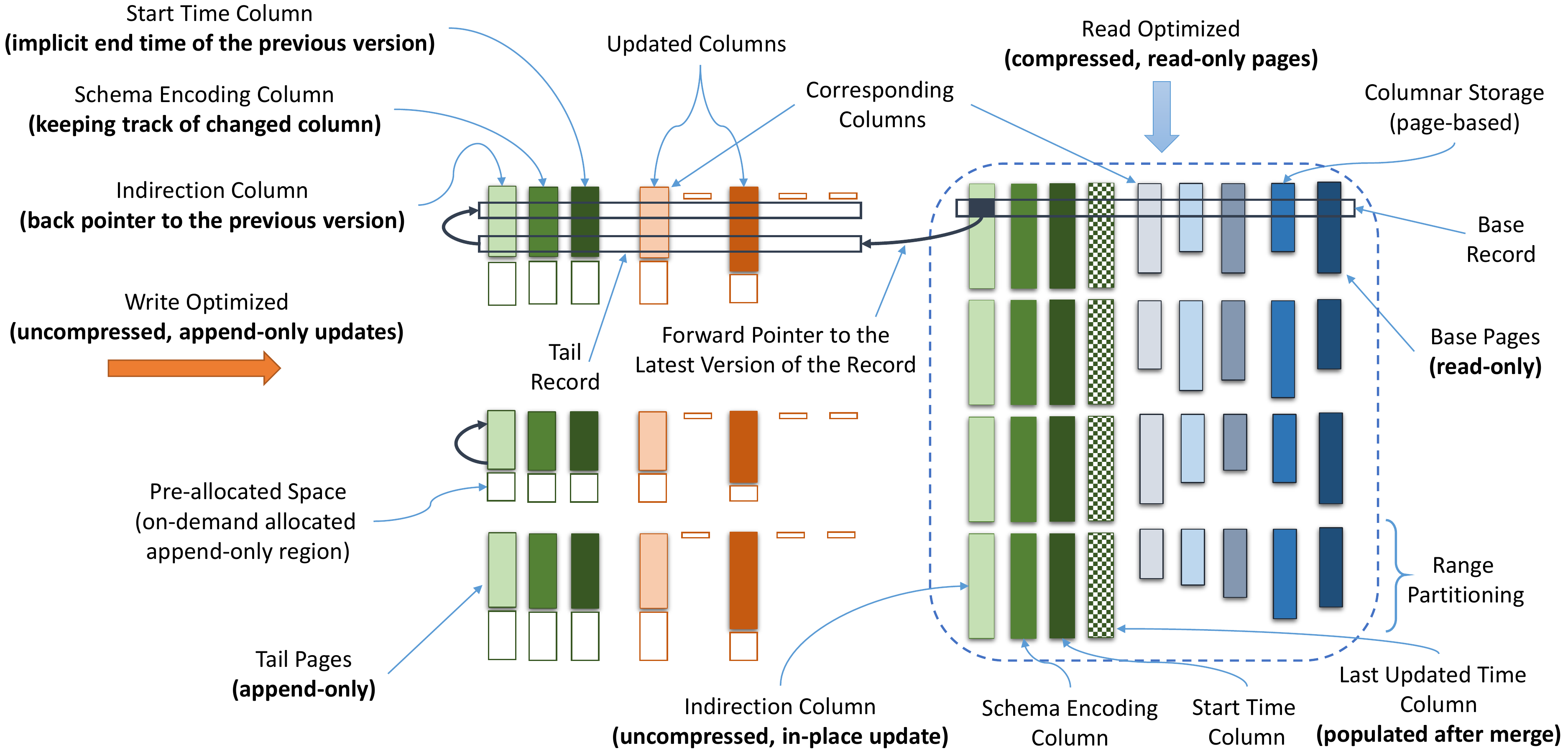}
\caption{Detailed, unfolded view of lineage-based storage architecture.}
\label{fig:unfolded}
\end{center}
\end{figure*}

\subsection{Lineage-based Storage Architecture}
\label{sec_arch}
In \tp, the storage layout is natively columnar that applies equally 
to both base and tail pages. A detailed view of our lineage-based storage architecture 
is presented in Figure~\ref{fig:unfolded}. In general, one can perceive 
tail pages as directly mirroring the structure and the schema of base 
pages. As we pointed out earlier, conceptually for every record, we distinguish 
between base vs. tail records, where each record is assigned a unique RID. 
But it is important to note that the RID assigned to a base record is stable 
and remains constant throughout the entire life-cycle of a record, and 
all indexes only reference base records (base RIDs); consequently, eliminating 
index maintenance problem associated when update operation results in creation 
of a new version of the record~\cite{Sadoghi:2013:MUD:2536222.2536226,indVLDJ}. 
When a reader performing index lookup, it always lands at a base record, and from the base 
record it can reach any desired version of the record by following the table-embedded 
indirection to access the latest (if the base record is out-of-date) or an earlier 
version of the record. However, when a record is updated, a new version is 
created. Thus, a new tail record is created to hold the new version, and 
the new tail record is assigned a new tail RID that is referenced by the base 
record (as demonstrated in Figure~\ref{fig:unfolded}). 

Each table in addition to having the standard data columns 
has several meta-data columns. These meta-data columns include the  
\textit{Indirection} column, the \textit{Schema Encoding} column, the 
\textit{Start Time} column, and the \textit{Last Updated Time} column. 
An example of table schema is shown in Table~\ref{tab:example}.

The \textit{Indirection} column exists in both the base and tail records. For base 
records, the \textit{Indirection} column is interpreted as a forward pointer to the 
latest version of a record residing in tail pages, essentially storing the RID 
of the latest version of a record. If a record has never been updated, then the 
\textit{Indirection} column will hold a null value. In contrast, for tail records, the 
\textit{Indirection} column is used to store a backward pointer to the last updated 
version of a record in tail pages. If no earlier version exists, then 
the \textit{Indirection} column will point to the RID of the base record.

The \textit{Schema Encoding} column stores the bitmap representation of 
the state of the data columns for each record, where there is one bit 
assigned for every column in the sch\-ema (excluding the meta-data columns), 
and if a column is updated, its corresponding bit in the \textit{Schema Encoding} 
column is set to 1, otherwise is set to 0. The schema encoding enables to 
quickly determine if a column has ever been updated or not (for base records) 
or to determine for each tail record, which columns have been updated and 
have explicit values as opposed to those columns that have not been updated 
and have an implicit special null values (denoted by $\varnothing$). An 
example of \textit{Schema Encoding} column is provided in Table~\ref{tab:example}.

The \textit{Start Time} column stores the time at which a base record 
was first installed in base pages (the original insertion time), and for a  
tail record, the \textit{Start Time} column holds the time at which the record was 
updated, which is also the implicit end time of the previous version of the 
record. The \textit{Start Time} column is essential for maintaining and distinguishing between 
different version of the record. In addition, to the \textit{Start Time} column, for base records, we maintain 
an optional \textit{Last Updated Time} column, which is only populated after 
the merge process is taken place and reflects the \textit{Start Time} of those tail records 
included in merged pages. Also note that the initial \textit{Start Time} column for  
base records is always preserved (even after the merge process) for faster pruning 
of those records that are not visible to readers because they fall outside the reader's snapshot.
Lastly, we may add the \textit{Base RID} column optionally to tail records to store 
the RIDs of their corresponding base records; this is utilized to improve the merge process. 
\textit{Base RID} is a highly compressible column that would require at most two bytes 
when restricting the range partitioning of records to $2^{16}$ records.

\begin{table*}[t]
\begin{center}
\begin{tabular}{| c | c | c | c | c | c | c | c | c |}  \hline
\textit{RID}	& \textit{Indirection}	& \textit{Schema Encoding} 	& \textit{Start Time} 	& \textit{Key}	& \textit{A}	& \textit{B}	& \textit{C} \\ [0.5ex] \hline\hline
\multicolumn{8}{|l|}{Partitioned base records for the key range of $k_1$ to $k_3$}  \\ \hline
$b_1$ 			& $t_8$					& 0000$\ \ $			& 10:02			& $k_1$	& $a_1$	& $b_1$	& $c_1$   \\ \hline
$b_2$ 			& $t_5$					& 0101$\ $			& 13:04			& $k_2$	& $a_2$	& $b_2$	& $c_2$   \\ \hline
$b_3$ 			& $t_7$					& 0001$\ $			& 15:05			& $k_3$	& $a_3$	& $b_3$	& $c_3$   \\ \hline \hline  
\multicolumn{8}{|l|}{Partitioned base records for the key range of $k_4$ to $k_6$}  \\ \hline
$b_4$ 			& $\bot$				& 0000$\ \ $			& 16:20			& $k_4$	& $a_4$	& $b_4$	& $c_4$   \\ \hline
$b_5$ 			& $\bot$				& 0000$\ \ $			& 17:21			& $k_5$	& $a_5$	& $b_5$	& $c_5$   \\ \hline
$b_6$ 			& $\bot$				& 0000$\ \ $			& 18:02			& $k_6$	& $a_6$	& $b_6$	& $c_6$   \\ \hline \hline 
\multicolumn{8}{|l|}{Partitioned tail records for the key range of $k_1$ to $k_3$}  \\ \hline
$t_1$ 			& $b_2$					& 0100*				& 13:04			& $\varnothing$	& $a_2$			& $\varnothing$	& $\varnothing$   \\ \hline
$t_2$ 			& $t_1$					& 0100$\ \ $			& 19:21			& $\varnothing$	& $a_{21}$		& $\varnothing$	& $\varnothing$   \\ \hline
$t_3$ 			& $t_2$					& 0100$\ \ $			& 19:24			& $\varnothing$	& $a_{22}$		& $\varnothing$	& $\varnothing$   \\ \hline
$t_4$ 			& $t_3$					& 0001*				& 13:04			& $\varnothing$	& $\varnothing$	& $\varnothing$	& $c_2$   	\\ \hline
{\color{red}$t_5$}&{\color{red}$t_4$}	&{\color{red}0101$\ \ $}&{\color{red}19:25}&{\color{red}$\varnothing$}&{\color{red}$a_{22}$}&{\color{red}$\varnothing$}	&{\color{red}$c_{21}$}  \\ \hline
$t_6$ 			& $b_3$					& 0001*				& 15:05			& $\varnothing$	& $\varnothing$	& $\varnothing$	& $c_3$   \\ \hline
{\color{red} $t_7$}&{\color{red} $t_6$}	&{\color{red} 0001$\ \ $}&{\color{red}19:45}& {\color{red}$\varnothing$}&{\color{red}$\varnothing$}&{\color{red}$\varnothing$}	&{\color{red}$c_{31}$}   \\ \hline
{\color{red} $t_8$}&{\color{red} $b_1$}	&{\color{red} 0000$\ \ $}&{\color{red}20:15}& {\color{red}$\varnothing$}&{\color{red}$\varnothing$}&{\color{red}$\varnothing$}	&{\color{red}$\varnothing$}   \\ \hline
\end{tabular}
\caption{An example of the update and delete procedures (conceptual tabular representation).}
\label{tab:example}
\end{center}
\end{table*}

\section{Fine-grained Storage Manipulation}
The transaction processing can be viewed as two major challenges: (1) how 
data is physically manipulated at the storage layer and how changes are propagated 
to indexes and (2) how multiple transactions (where each transaction consists of 
many statements) can concurrently coordinate reading and writing of the shared data. 
The focus of this paper is on the former challenge, and we defer the latter to our 
discussion on the employed concurrency model in Section~\ref{sec_rec}.


\subsection{Update and Delete Procedures}
\label{sec_update}
Without the loss of generality, from the perspective of the storage layer,  we focus 
on how to handle a single point update or delete in \tp\ (but note that we support 
multi-statement transactions through \tp's transaction layer as demonstrated by 
our evaluation). Each update may affect a single or multiple records. Since records are 
(virtually) partitioned into a set of disjoint ranges (as shown in Table~\ref{tab:example}), 
each updated record naturally falls within only one range. Now for each range of records, 
upon the first update to that range, a set of tail pages are created (and persisted on 
disk optionally) for the updated columns and are added to the page directory, i.e., lazy 
tail-page allocation. Consequently, updates for each record range are appended to their 
corresponding tail pages of the updated columns only; thereby, retraining all versions of 
the record, avoiding in-place updates of modified data columns, and clustering updates 
for a range of records within their corresponding tail pages.

To describe the update procedure in \tp, we rely on our running example shown in 
Table~\ref{tab:example}. When a transaction updates any column of a record for 
the first time, two new tail records (each tail record is assigned a unique RID) 
are created and appended to the corresponding tail pages. For example, consider 
updating the column \textit{A} of the record with the key $k_2$ (referenced by 
the RID $b_2$) in Table~\ref{tab:example}. The first tail record, referenced by 
the RID $t_1$, contains the original value of the updated column, i.e., $a_2$, 
whereas implicit null values ($\varnothing$) are preassigned for remaining unchanged 
columns. Taking a snapshot of the original changed values becomes essential in order 
to ensure contention-free merging as discussed in Section~\ref{sec_merge}. 
The second tail record contains the newly updated value for column \textit{A}, 
namely, $a_{21}$, and again implicit special null values for the rest of the columns;  
a column that has never been updated does not even have to be materialized with 
special null values. However, for any subsequent updates, only one tail record is 
created, e.g., the tail record $t_3$ is appended as a result of updating the 
column \textit{A} from $a_{21}$ to $a_{22}$ for the record $b_2$. 

In general, updates could either be cumulative or non-cumulative. The cumulative 
property implies that when creating a new tail record, the new record will contain 
the latest values for all of the updated columns thus far. For example, consider 
updating the column \textit{C} for the record $b_2$. Since the column \textit{C} of  
the record $b_2$ is being updated for the first time, we first take a snapshot of its old 
value as captured by the tail record $t_4$. Now for the cumulative update, a new tail 
record is appended that repeats the previously updated column \textit{A}, as demonstrated 
by the tail record $t_5$. If non-cumulative update approach was employed, then the tail 
record would consists of only the changed value for column \textit{C} and not \textit{A}. 
It is important to note that cumulation of updates can be reset at anytime. In the
absence of cumulation, readers are simply forced to walk back the chain of recent 
versions to retrieve the latest values of all desired columns. Thus, cumulative 
update is an optimization that is intended to improve the read performance. 

\begin{figure*}[t]
\begin{center}
\includegraphics[scale=0.415]{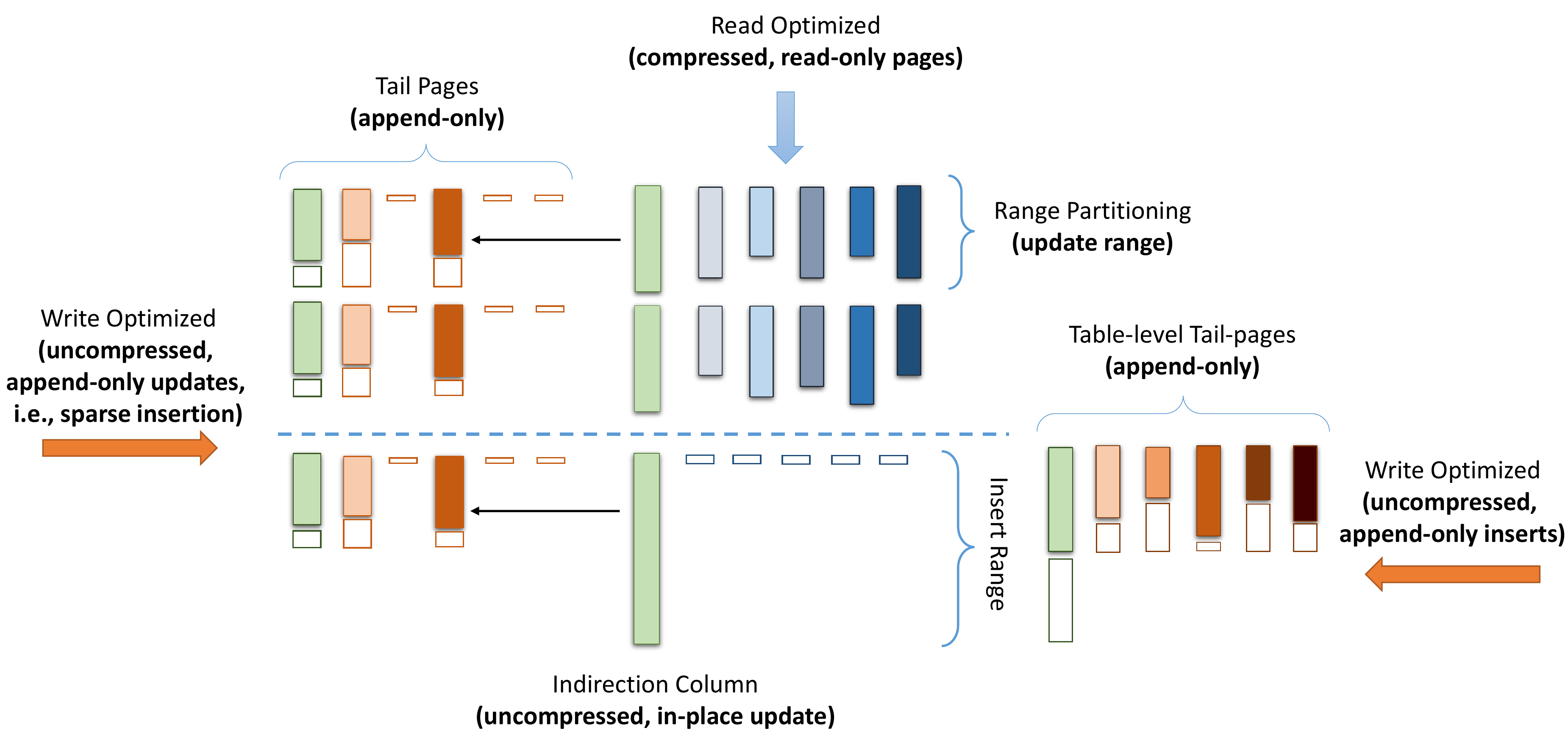}
\caption{Append-only insertion of new records with concurrent updates (by employing tail pages).}
\label{fig:insert}
\end{center}
\end{figure*}

As part of the update routine, the embedded \textit{Indirection} column (forward 
pointers) for base records is also updated to point to the newly created 
tail record. In our running example, the \textit{Indirection} 
column of the record $b_2$ points to the tail record $t_5$. Also after updating 
the column \textit{C} of the record $b_3$, the \textit{Indirection} column points 
to the latest version of $b_3$, which is given by $t_7$. Likewise, the 
\textit{Indirection} column in the tail records point to the previous  
version of the record. It is important to note that the \textit{Indirection} column of 
base records is the only column that requires an in-place update in our architecture. However, 
as discussed in our low-level synchronization protocol (cf. Section \ref{sec_rec}), this is a 
special column that lends itself to latch-free concurrency protocol. 

Furthermore, indexes always point to base records (i.e., base RIDs), and they are never 
directly point to any tail records (i.e., tail RIDs) in order to avoid the index maintenance 
cost that arise in the absence of in-place update mechanism~\cite{Sadoghi:2013:MUD:2536222.2536226,indVLDJ}. %
Therefore, when a new version of a record is created (i.e., a new tail record), first, all 
indexes defined on unaffected columns do not have to be modified and, second, only the 
affected indexes are modified with the updated values, but they continue to point to 
base records and not the newly created tail records. 
Suppose there is an index defined on the column \textit{C} (cf. Table~\ref{tab:example}). 
Now after modifying the record $b_2$ from $c_2$ to $c_{21}$, we add the new entry 
$(c_{21}, b_2)$ to the index on the column \textit{C}.\footnote{Optionally the old 
value $(c_{2}, b_2)$ could be removed from the index; however, its removal may affect 
those queries that are using indexes to compute answers under snapshot semantics. Therefore, 
we advocate deferring the removal of changed values from indexes until the changed entries 
fall outside the snapshot of all relevant active queries.} Subsequently, when a reader looks up the 
value $c_{21}$ from the index, it always arrives at the base record $b_2$ initially, then 
the reader must determine the visible version of $b_2$ (by following the indirection if 
necessary) and must check if the visible version has the value $c_{21}$ for the column 
\textit{C}, essentially re-evaluating the query predicates.

There are two other meta-data columns that are affected by the update procedure. 
The \textit{Start Time} column for tail records simply holds the time at which 
the record was updated (an implicit end of the previous version). For example, 
the record $t_7$ has a start time of 19:45, which also implies that the end time of 
the first version of the record $b_3$. The \textit{Schema Encoding} column is a concise 
representation that shows which data columns have been updated thus far. For example, 
the \textit{Schema Encoding} of the tail record $t_7$ is set to ``0001'', which 
implies that only the column \textit{C} has been changed. To distinguish between 
whether a tail record is holding new values or it is the snapshot of old values, we add 
a flag to the \textit{Schema Encoding} column, which is shown as an asterisk. For 
example, the tail record $t_6$ stores the old value of the column \textit{C}, which 
is why its \textit{Schema Encoding} is set to ``0001\textbf{*}''. The \textit{Schema 
Encoding} can also be maintained optionally for base records as part of the 
update process or it could be populated only during the merge process.

Notably, when there are multiple individual updates to the same record by the same 
transaction, then each update is written as a separate entry to tail pages. Each 
update results in a creation of a new tail record and only the final update becomes 
visible to other transactions. The prior entries are implicitly invalidated 
and skipped by readers. Also delete operation is simply translated into an update 
operation, in which all data columns are implicitly set to $\varnothing$, e.g., 
deleting the record $b_1$ results in creating the tail record $t_8$. An alternative 
design for delete is to create a tail record that holds a complete snapshot of the 
latest version of the deleted record.

\begin{table*}[t]
\begin{center}
\begin{tabular}{| c | c | c | c | c | c | c | c | c |}  \hline
\textit{RID}	& \textit{Indirection}	& \textit{Schema Encoding} 	& \textit{Start Time} 	& \textit{Key}	& \textit{A}	& \textit{B}	& \textit{C} \\ [0.5ex] \hline\hline
\multicolumn{8}{|l|}{Partitioned base records for the key range of $k_4$ to $k_6$}  \\ \hline
$b_4$ 			& $\bot$				& 0000$\ \ $			& 16:20			& $k_4$	& $a_4$	& $b_4$	& $c_4$   \\ \hline
$b_5$ 			& $\bot$				& 0000$\ \ $			& 17:21			& $k_5$	& $a_5$	& $b_5$	& $c_5$   \\ \hline
$b_6$ 			& $\bot$				& 0000$\ \ $			& 18:02			& $k_6$	& $a_6$	& $b_6$	& $c_6$   \\ \hline \hline 
\multicolumn{8}{|l|}{Insert range for the base record with the base RID range of $b_7$ to $b_9$}  \\ \hline
$b_7$ 			& $\bot$					& \multicolumn{6}{c|}{} \\ \hline
$b_8$ 			& $t_{14}$					& \multicolumn{6}{c|}{} \\ \hline
$b_9$ 			& $t_{16}$					& \multicolumn{6}{c|}{} \\ \hline
\multicolumn{8}{|l|}{Table-level tail-pages for the base record with the base RID range of $b_7$ to $b_9$}  \\ \hline
$tt_7$ 			& $t_7$					& 0000$\ \ $			& 18:30			& $k_7$	& $a_7$	& $b_7$	& $c_7$   \\ \hline
$tt_8$ 			& $t_8$					& 0000$\ \ $			& 18:45			& $k_8$	& $a_8$	& $b_8$	& $c_8$   \\ \hline
$tt_9$ 			& $t_9$					& 0000$\ \ $			& 19:05			& $k_9$	& $a_9$	& $b_9$	& $c_9$   \\ \hline
\multicolumn{8}{|l|}{Partitioned tail records for the key range of $k_7$ to $k_9$}  \\ \hline
$t_{13}$ 			& $b_8$					& 0001*				& 18:45			& $\varnothing$	& $\varnothing$	& $\varnothing$	& $c_8$   	\\ \hline
{\color{red}$t_{14}$}&{\color{red}$t_{13}$}	&{\color{red}0001$\ \ $}&{\color{red}22:25}&{\color{red}$\varnothing$}&{\color{red}$\varnothing$}&{\color{red}$\varnothing$}	&{\color{red}$c_{81}$}  \\ \hline
$t_{15}$ 			& $b_3$					& 0100*				& 19:05			& $\varnothing$	& $a_9$	& $\varnothing$	& $\varnothing$   \\ \hline
{\color{red} $t_{16}$}&{\color{red} $t_{15}$}	&{\color{red} 0100$\ \ $}&{\color{red}22:45}& {\color{red}$\varnothing$}&{\color{red}$a_{91}$}&{\color{red}$\varnothing$}	&{\color{red}$\varnothing$}   \\ \hline
\end{tabular}
\caption{An example of insertion with concurrent updates (conceptual tabular representation).}
\label{tab:example_insert}
\end{center}
\end{table*}

\subsection{Insert Procedure}
\label{sec:insert}
In OLTP workloads, another important fine-grained manipulation is the insertion 
of new records. Conceptually, the table naturally grows by inserting new records 
to the end of the table (append-only mechanism). We rely on a simpler manifestation 
of our notion of tail pages followed by the transformation of tail pages into compressed, 
read-only base pages through a simplified merge process. In fact, one can 
even view our previously described update mechanism as a form of sparse 
insertion. 

In our proposed insert design, we designate the end of the table as the insert 
range. An insert range is basically a pre-allocated range of base RIDs for 
accommodating future insertions. In practice, the insert range size (at least 
a million RIDs) is much larger than our range partitioning that is 
employed for update processing (i.e., update range).\footnote{Each table may have more 
than one insert range to support a higher degree of concurrency if the workload is 
insert intensive.}. For the insert range, we allocated a set of tail pages for 
appending new records, which we refer to them as ``\textit{table-level tail-pages}'' 
even though structurally there is no difference between table-level tail-pages 
vs. regular tail pages. Figure~\ref{fig:insert} pictorially captures our insert 
design. In table-level tail-pages, we allocate tail pages for all columns 
(unlike for updates that was limited to only the updated columns) because the 
insert statement always provide a value for every column (even if it is an 
implicit null value for a nullable column).

Adding a new insert range consists of reserving a set of base RIDs (e.g., in the 
order of millions) and a set of tail RIDs; these two sets of RIDs are equal in size and aligned. 
Thus, the 10$^{th}$ base RID in the insert range corresponds to the 10$^{th}$ 
tail RID in the table-level tail-range (i.e., both ranges following the same 
insertion order). The alignment of RIDs allows implicit addressing for looking 
up a record in the insert range. When a new record is about to be inserted to the 
table, the new record receives a reserved base RID in the insert range and 
the corresponding tail RID in the table-level tail-range. If insert range is full, 
then a new insert range is created. But the key guiding principle for insertion 
is to satisfy the \textit{stability property} of the base pages (i.e., read-only) 
with the exception of the \textit{Indirection} column that is updated in-place. Therefore, 
the insertion procedure simply consists of acquiring base and tail RIDs, insert the actual 
record to table-level tail-pages, and setting the \textit{Indirection} column in the 
base record to null. Alternatively, the \textit{Indirection} column could be set to null 
when allocating pages for the insert range.

An example of insertion is illustrated in Table~\ref{tab:example_insert}. The insert range 
is shown as $b_7$ to $b_9$, and the table-level tail-range is shown as $tt_7$ to $tt_9$. The 
first inserted record is $(k_7, a_7, b_7, c_7)$ with the key $k_7$ that is assigned $b_7$ 
as its base RID and $tt_7$ as its tail RID. The only column allocated for base records  
is the \textit{Indirection} column, which is initially set to null ($\bot$). The actual 
values for the meta-data and data columns are appended to the table-level tail-pages at 
the position given by the tail RID $tt_7$. In the same spirit, the records with $b_8$ and 
$b_9$ are also appended to the insert range. Now if a recently inserted record is updated, 
then the update follows the same path as explained earlier (cf. Section~\ref{sec_update}). 
Suppose the record $b_8$ is updated by modifying the value of its \textit{C} column from 
$c_8$ to $c_{81}$. The update simply results in acquiring a new tail RID in the regular 
tail pages (as before) and appending only the updated column followed by updating the 
\textit{Indirection} column in-place. This is demonstrated by appending the tail record 
$t_{14}$ to the corresponding tail pages and setting the \textit{Indirection} column of 
the record $b_8$ to $t_{14}$.

\begin{figure*}[t!]
	\begin{minipage}[t]{0.5\linewidth}		
        \centering
        \includegraphics[scale=0.32]{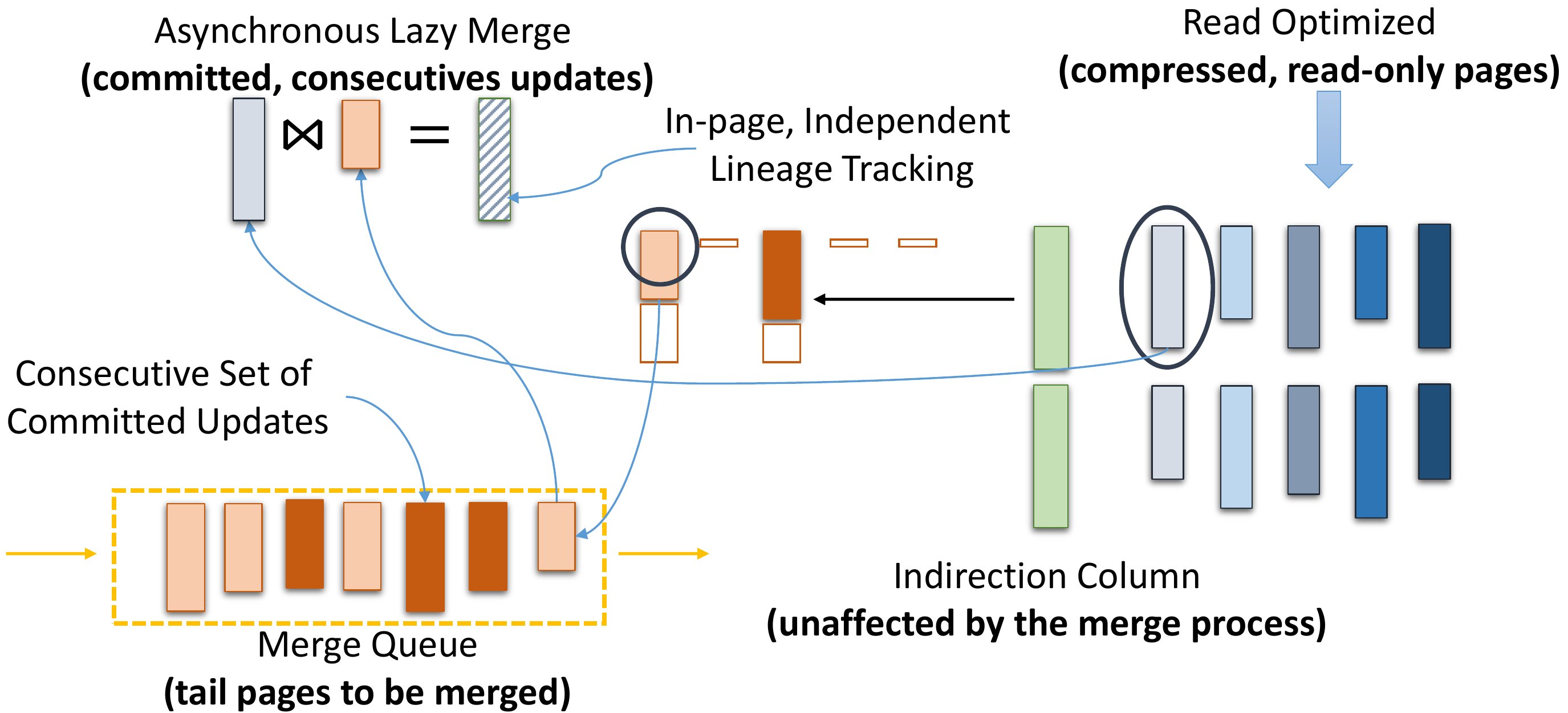}
        \caption{Lazily, independently merging of tail \& base pages.}
		\label{fig:merge}
	\end{minipage}\hfill
	\hspace{5mm}
	\begin{minipage}[t]{0.5\linewidth}
		\includegraphics[scale=0.35]{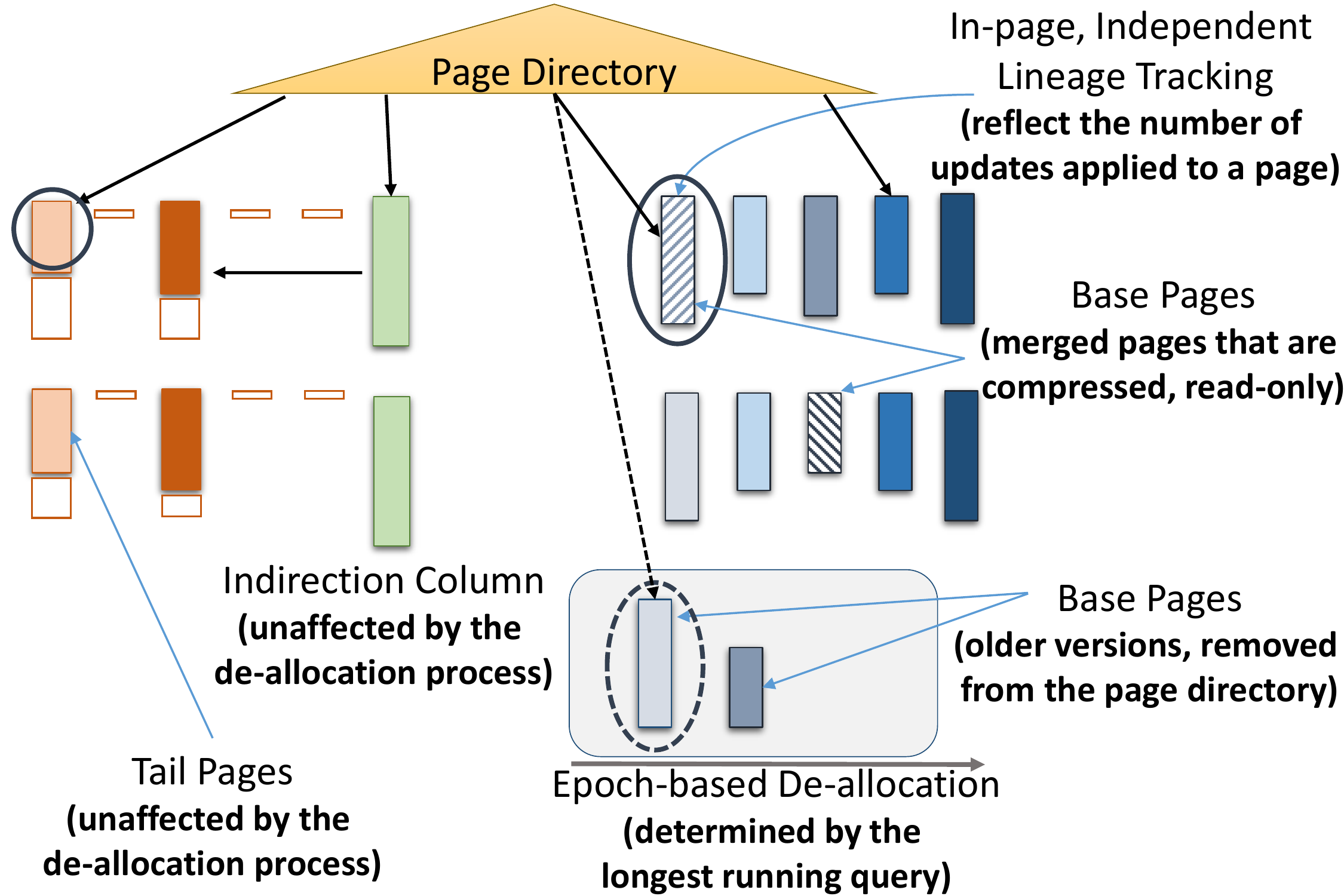}
		\caption{Epoch-based, contention-free page de-allocation.}
		\label{fig:deallocation}
	\end{minipage}\hfill
	\hspace{-3mm}
\end{figure*}

\section{Real-time Storage Adaption}
To ensure a near optimal storage layout, outdated base pages are merged 
lazily with their corresponding tail pages in order to preserve the efficiency of 
analytical query processing. Recall that the base pages are read-only and compressed 
(read optimized) while the tail pages are uncompressed\footnote{Even though compression 
techniques such as local and global dictionaries can be employed in tail pages, but these 
directions are outside the scope of the current work.} that grow using a strictly 
append-only technique (write optimized). Therefore, it is necessary to transform the 
recent committed updates (accumulated in tail pages) that are write optimized into 
read optimized form. A distinguishing feature of our lineage-based architecture is to 
introduce a contention-free merging process that is carried out completely in the background 
without interfering with foreground transactions. Furthermore, the contention-free merging 
procedure is applied only to the updated columns of the affected update ranges. 
There is even no dependency among columns during the merge; thus, the different columns 
of the same record can be merged completely independent of each other at different points  
in time. This is achieved by independently maintaining in-page lineage information for each 
merged page. The merge process is conceptually depicted in Figure~\ref{fig:merge}, in which 
writer threads (i.e., update transactions) place candidate tail pages to be merged into 
the merge queue while the merge thread continuously takes pages from the queue and 
processes them.

\comment{
\begin{figure}
\begin{center}
\includegraphics[scale=0.31]{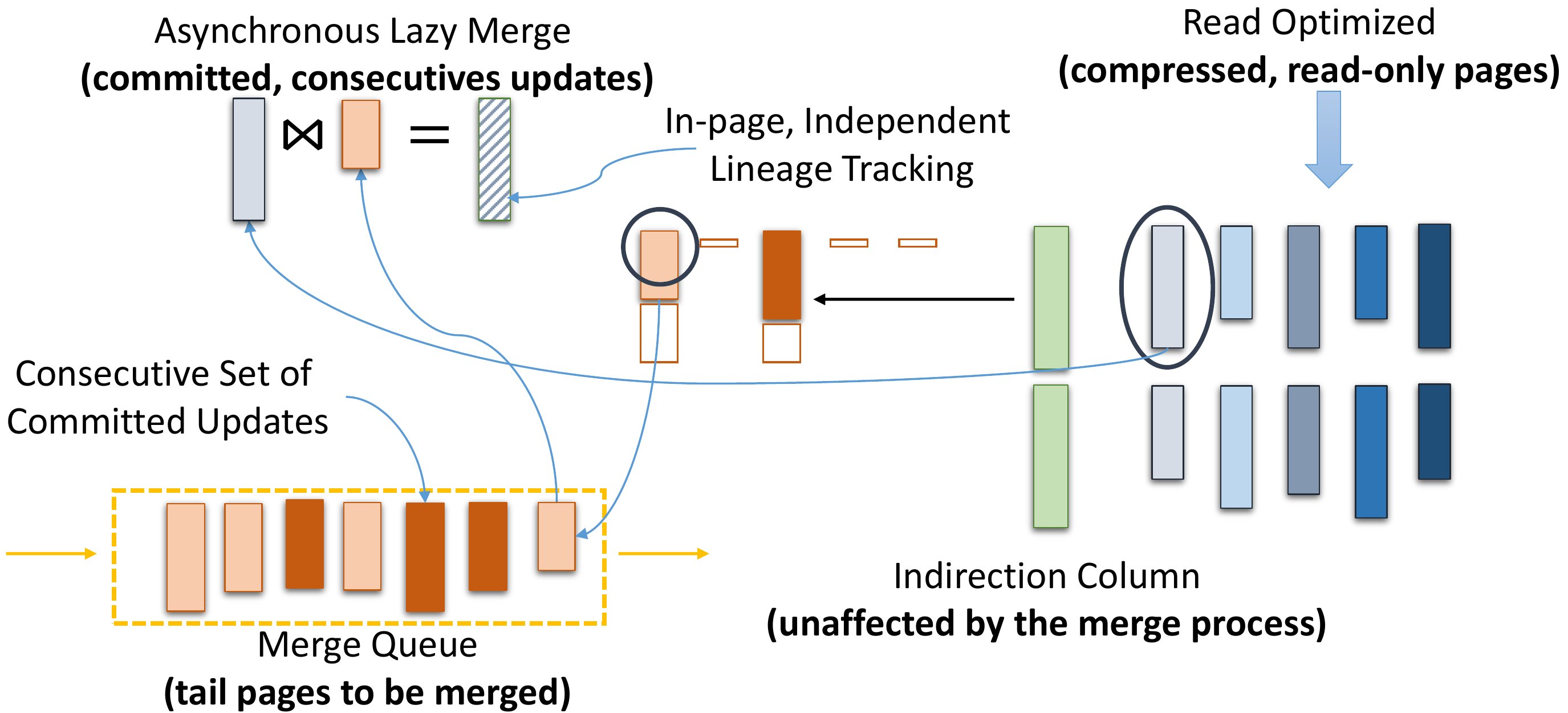}
\caption{Lazily, independently merging of tail \& base pages.}
\label{fig:merge}
\end{center}
\end{figure}
}

\subsection{Contention-free, Relaxed Merge}
\label{sec_merge}
In \tp, we abide to one \textit{main design principle} for ensuring 
con\-tent\-ion-free processing that is \textit{``always operating on stable data''}. 
The inputs to the merge process are (1) a set of base pages (committed base records) 
that are read-only,\footnote{The \textit{Indirection} column is the only column that undergoes 
in-place update that also never participates in the merge process.} thus, stable data 
and (2) a set of consecutive committed tail records in tail pages,\footnote{Note that 
not every committed update has to be applied as the merge process is relaxed, and 
the merge eventually process all committed tail records.} thus, also stable data. 
The output of the merge process (that is also relaxed) is a set of newly consolidated 
base pages (also referred to as merged pages) with in-page lineage information that are 
read-only, compressed, and almost up-to-date, thus, stable data. To decouple users' transactions 
(writers) from the merge process, we also ensure that the write path of the ongoing transactions 
does not overlap with the write path of the merge process. Writers append new uncommitted 
tail records to tail pages (but as stated before uncommitted records do not participate 
in the merge), and writers perform in-place update of the \textit{Indirection} 
column within base records to point to the latest version of the updated records 
in tail pages (but the \textit{Indirection} column is not modified by the merge 
process), whereas the write path of the merge process consists of creating only a 
new set of read-only base pages. 

\comment{
\begin{figure}[t]
\begin{center}
\includegraphics[scale=0.36]{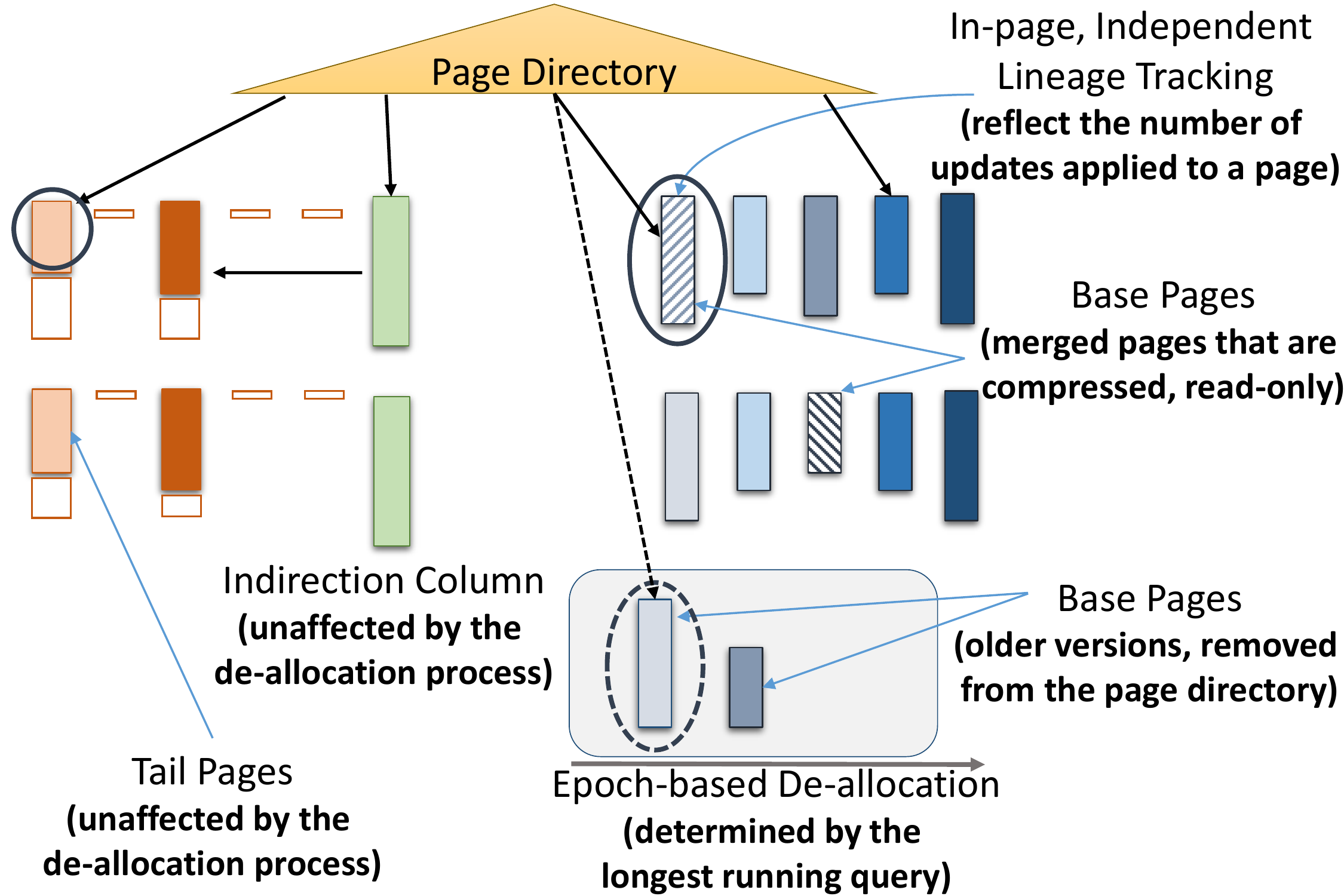}
\caption{Epoch-based, contention-free page de-allocation.}
\label{fig:deallocation}
\end{center}
\end{figure}
}

\subsubsection{Merge Algorithm}
\label{sec_merge_algorithm}
The details of the merge algorithm, conceptually resembling the standard left-outer join, 
consists of (1) identifying a set of committed tail records in tail pages; (2) 
loading the corresponding outdated base pages; (3) consolidating the base and 
tail pages while maintaining the in-page lineage; (4) updating the page directory; and (5) de-allocating 
the outdated base pages. The pseudo code for the merge is shown in Algorithm~\ref{alg:merge-algo}, 
where each of the five mentioned steps are also highlighted.

\begin{steps}

\textbf{$\ \ $Step 1: Identify committed tail records in tail pages:} 
Select a set of consecutive fully committed tail records (or pages) since the last 
merge within each update range. 

\textbf{Step 2: Load the corresponding outdated base pages:} 
For a selected set of committed tail records, load the corresponding outdated 
base pages for the given update range (limit the load to only outdated columns). 
This step can further be optimized by avoiding to load sub-ranges of records that 
have not yet changed since the last merge. No latching is required when loading the 
base pages.

\textbf{Step 3: Consolidate the base and tail pages:} 
For every updated column, the merge process will read $n$ outdated base pages 
and applies a set of recent committed updates from the tail pages and writes out $m$ new 
pages.\footnote{At most up to one merged page per column could be left underutilized 
for a range of records after the merge process. To further reduce the underutilized 
merged pages, one may define finer range partitioning for updates (e.g., $2^{12}$ 
records), but operate merges at coarser granularity (e.g., $2^{16}$ records). This 
will provide the benefit of locality of access for readers given smaller range size 
of $2^{12}$, yet it provides a better space utilization and compression for newly 
created merge pages when larger ranges are chosen (cf. Section~\ref{sec_range}).} 
First the \textit{Base RID} column of the committed tail pages (from Step 1) 
are scanned in reverse order to find the list of the latest version of every 
updated record since the last merge (a temporary hashtable may be used to keep 
track whether the latest version of a record is seen or not). Subsequently,  
applying the latest tail records in a reverse order to the base records 
until an update to every record in the base range is seen or the list is 
exhausted (skip any intermediate versions for which a newer update exists in 
the selected tail records). If a latest tail record indicates the deletion of the record, 
then the deleted record will be included in the consolidated records.\footnote{Alternatively, 
if all the deleted values are also stored in tail records, then it is sufficient 
to fill all data columns with the special null value $\varnothing$ for deleted 
records in the final merged pages. However, we would still need to preserve the 
\textit{Indirection} column of deleted records in order to provide access 
to the earlier versions of deleted records.} The merged pages will keep track of the lineage 
information in-page, i.e., tracking how many tail records have been consolidated thus far.  
Any compression algorithm (e.g., dictionary encoding) can be applied on the consolidated pages 
(on column basis) followed by writing the compressed pages into newly created pages. Moreover, 
the old \textit{Start Time} column is remained intact during the merge process because this column 
is needed to hold the original insertion time of the record.\footnote{The \textit{Start Time} column 
is also highly compressible column with a negligible space overhead to maintain it.} 
Therefore, to keep track of the time for the consolidated records, the 
\textit{Last Updated Time} column is populated to store the Start Time of the applied 
tail records. The \textit{Schema Encoding} column may also be populated during 
the merge to reflect all the columns that have been changed for each record.

\textbf{Step 4: Update the page directory:} 
The pointers in the page directory are updated to point to the newly created 
merged pages. Essentially this is the only foreground action taken by the merge 
process, which is simply to swap and update pointers in the page directory -- 
an index structure that is updated rarely only when new pages are allocated. 

\textbf{Step 5: De-allocate the outdated base pages:} 
The outdated base pages are de-allocated once the current readers are drained naturally 
via an epoch-based approach. The epoch is defined as a time window, in which the 
outdated base pages must be kept around as long as there is an active query that started 
before the merge process. Pointers to the outdated base pages are kept in a queue to be 
re-claimed at the end of the query-driven epoch-window. The pointer swapping and 
the page de-allocation are illustrated in Figure~\ref{fig:deallocation}. $\blacksquare$

\end{steps}

\begin{algorithm}[!t]
\DontPrintSemicolon
\SetKwInOut{input}{Input}\SetKwInOut{output}{Output}
{\small
\input{Queue of unmerged committed tail pages (mergeQ)}
\output{Queue of outdated and consolidated base pages to be deallocated (deallocateQ)}
    \While {true} {
     \mbox{// \textit{\textbf{Step 1}}}\;
     \mbox{// wait until the the concurrent merge queue is not empty}\;
        \If {mergeQ is not empty} {
            \mbox{// \textit{\textbf{Step 2}}}\;
            \mbox{// fetch references to a set of committed tail pages}\;
            \mbox{batchTailPage $\dashleftarrow$ mergeQ.dequeue()}\;
            \mbox{// create a copy of corresponding base pages}\;         
            \mbox{batchConsPage $\leftarrow$ batchTailPage.getBasePageCopy()}\;    
			 \mbox{decompress(batchConsPage)}\;    
		\mbox{// track if it has seen the latest update of every record}\;
        \mbox{HashMap seenUpdatesH}\;        
        \mbox{//reading a set of tail pages in reverse order}\;
        \mbox{// \textit{\textbf{Step 3}}}\;
        \For {$i = 0$; $i < $batchTailPage.size; $i \gets i + 1$} {
        	\mbox{tailPage $\gets$ batchTailPages[i]}\;
	        \For {$j = k-1$; $j \geq $ tailPage.size; $j \gets j - 1$} {
				\mbox{record[j]  $\leftarrow$ $j^{th}$ record in the tailPage}\;
				\mbox{RID $\leftarrow$ record[j].RID}\;
				\If {seenUpdatesH does not contain RID} {
						\mbox{seenUpdatesH.add(RID)}\;
						\mbox{// copy the latest version of record}\;
						\mbox{// $ \ \ \ \ \ \ \ \ \ \ \ \ $ into consolidated pages}\;
						\mbox{batchConsPage.update(RID, record[j])}\;
                }
				\If {if all RIDs OR all tail pages are seen} {
					\mbox{compress(batchConsPage)}\;
					\mbox{persist(batchConsPage)}\;
					\mbox{stop examining remaining tail pages}\;
				}
            }
        }
        \mbox{// \textit{\textbf{Step 4}}}\;
        \mbox{// fetch references to the corresponding base pages}\;
		\mbox{batchBasePage $\dashleftarrow$ batchTailPage.getBasePageRef()}\;       		
        \mbox{// update page directory to point}\;
        \mbox{// $ \ \ \ \ \ \ \ \ \ \ \ \ $ to the consolidated base pages}\;
	    \mbox{PageDirect.swap(batchBasePage, batchConsPage)}\;	   
	    \mbox{// \textit{\textbf{Step 5}}}\;
        \mbox{// push to queue to deallocate the outdated pages}\;
        \mbox{// $ \ \  $once the readers  before merge are drained}\;
         deallocateQ.enqueue(batchBasePage)\;
         }
   }            
}
\caption{Merge Algorithm}
\label{alg:merge-algo}
\end{algorithm}

An example of our merge process is shown in Table~\ref{tab:example_merge} based 
on our earlier update example, in which we consolidate the first seven tail records 
(denoted by $t_1$ to $t_7$) with their corresponding base pages. The resulting 
merged pages are shown, where the affected records are highlighted. Note that 
only the updated columns are affected by the merge process (and the \textit{Indirection} 
column is not affected). Furthermore, not all updates are needed to be applied, 
only the latest version of every updated record needs to be consolidated 
while the other entries are simply discarded. In our example, only the tail 
records $t_5$ and $t_7$ participated in the merge, and the rest were discarded.

\begin{table*}[t]
\begin{center}
\begin{tabular}{| c | c | c | c | c | c | c | c | c | c |}  \hline
\textit{RID}	& \textit{Indirection}	& \makecell{\textit{Schema} \\ \textit{Encoding}} 	& \makecell{\textit{Start} \\ \textit{Time}} & \makecell{\textit{Last} \\ \textit{Updated} \\ \textit{Time}}	& \textit{Key}	& \textit{A}	& \textit{B}	& \textit{C} \\ [0.5ex] \hline\hline
\multicolumn{9}{|l|}{Partitioned base records for the key range of $k_1$ to $k_3$; Tail-page Sequence Number (TPS) = 0}  \\ \hline
$b_1$ 			& $t_8$				& 0000$\ \ $		& 10:02 &			& $k_1$	& $a_1$	& $b_1$	& $c_1$   \\ \hline
$b_2$ 			& $t_5$					& 0101$\ $			& 13:04	&		& $k_2$	& $a_2$	& $b_2$	& $c_2$   \\ \hline
$b_3$ 			& $t_7$					& 0001$\ $			& 15:05	&		& $k_3$	& $a_3$	& $b_3$	& $c_3$   \\ \hline \hline  
\multicolumn{9}{|l|}{Relevant tail records (below TPS $\leq t_7$ high-watermark) for the key range of $k_1$ to $k_3$}  \\ \hline
{\color{red}$t_5$}&{\color{red}$t_4$}	&{\color{red}0101$\ \ $}&{\color{red}19:25}&&{\color{red}$\varnothing$}&{\color{red}$a_{22}$}&{\color{red}$\varnothing$}	&{\color{red}$c_{21}$}  \\ \hline
{\color{red} $t_7$}&{\color{red} $t_6$}	&{\color{red} 0001$\ \ $}&{\color{red}19:45}&& {\color{red}$\varnothing$}&{\color{red}$\varnothing$}&{\color{red}$\varnothing$}	&{\color{red}$c_{31}$}   \\ \hline
\multicolumn{9}{|l|}{Resulting merged records for the key range of $k_1$ to $k_3$; TPS = $t_7$}  \\ \hline
$b_1$ 			& $t_8$				& 0000$\ \ $	& 10:02 & 10:02			& $k_1$	& $a_1$	& $b_1$	& $c_1$   \\ \hline
$b_2$ 			& $t_5$					&{\color{red}0101$\ \ $}& 13:04&{\color{red}19:25}& $k_2$ &{\color{red}$a_{22}$}& $b_2$	&{\color{red}$c_{21}$}  \\ \hline
$b_3$ 			& $t_7$					& {\color{red} 0001$\ \ $}& 15:05& {\color{red}19:45} & $k_3$	& $a_3$	& $b_3$	& {\color{red}$c_{31}$}   \\ \hline
\end{tabular}
\caption{An example of the relaxed and almost up-to-date merge procedure (conceptual tabular representation).}
\label{tab:example_merge}
\end{center}
\end{table*}

\textbf{Merging Table-level Tail-pages:} 
It is important to note that merging table-level tail-pages with base 
pages in the insert range follows a similar process as above with a few 
simplification. First, the consolidation process is rather trivial because 
tail records in the table-level tail-range follows the exact same insertion 
order in the insert range (a trivial join-like operation). Also the insert range 
does not actually have any value except for the \textit{Indirection} column (which does 
not even participate in the merge itself). Thus, the merge process is essentially reading 
a set of consecutive committed tail records and compressing them to create 
a set of newly merged pages. Another simplification is that after the merged 
pages are created and the page directory is updated, then the old table-level 
tail-pages can be discarded permanently after all the active queries that 
started prior to the merge process are terminated. In contrast, the regular 
tail pages survive after the merge in order to enable answering historic 
queries and to avoid interfering with update transactions. All base records 
that have been merged with their table-level tail-pages are considered to be 
outside the insert range.

We further strengthen our \textit{data stability} condition for bringing base 
pages up-to-date. Earlier we stated that the merge only operates on a set of 
committed consecutive tail records, but no condition was imposed on the base 
records. Now we strengthen this condition by requiring that the base records 
must also fall outside the insert range before becoming a candidate for merging 
the recent updates.

\subsubsection{Merge Correctness Analysis}
A key distinguishing feature of our lineage-based storage architecture is 
to allow contention-free merging of tail and base pages without interfering 
with concurrent transactions. To formalize our merge process, we prove that 
merge operates only on stable data while maintaining in-page lineage without 
any information loss and that the merge does not limit users' transactions to access 
and/or modify the data that is being merged.

\begin{lemma}
Merge operates strictly on stable data.
\label{lem_stable}
\end{lemma}

\begin{proof}
Recall that by construction, we enforced that merge \textit{``always operate on 
stable data''}. The inputs to the merge process are (1) a set of base pages 
consisting of committed base records that are read-only\footnote{We also consider 
only base records that are outside the insert range, for additional details please 
refer to Section~\ref{sec:insert}.}, thus, stable data and (2) a set of consecutive 
committed tail records in tail pages, thus, also stable data. The output of the merge 
process is a set of newly merged pages that are read-only, thus, stable data as well. 
Hence, the merge process strictly takes as inputs stable data and produces 
stable data as well. \qed
\end{proof}

\begin{lemma}
Merge safely discards outdated base pages without violating any query's snapshot.
\label{lem_safe}
\end{lemma}

\begin{proof}
In order to support snapshot isolation semantics and time travel queries, 
we need to ensure that earlier versions of records that participate in the merge 
process are retained. Since we never perform in-place updates and each update is  
transformed into appending a new version of the record to tail pages, then as 
long as tail pages are not removed, we can ensure that we have access to every 
updated version. But recall that outdated base pages are de-allocated using our 
proposed epoch-based approach after being merged. Also note that base pages contain 
the original values of when a record was first created. Therefore, any original 
values that later were updated must be stored before discarding outdated base pages 
after a merge is taken place. In another words, we must ensure that outdated base 
pages are discarded safely.

As a result, the two fundamental criteria, namely, relaxing the merge (i.e. 
constructing an almost up-to-date snapshot) and operating on stable data, are not 
sufficient to ensure the \textit{safety property} of the merge. The last missing 
piece that enables safety of the merge is accomplished by taking a snapshot of the 
original values when a column is being updated for the first time (as described in 
Section~\ref{sec_update}). In other words, we have further strengthened our 
\textit{data stability} criterion by ensuring even \textit{stability in the committed history}. 
Hence, outdated base pages can be safely discarded without any information loss, 
namely, the merge process is safe. \qed
\end{proof}

\begin{theorem}
The merge process and users' transactions do not contend for base and tail pages 
or the resulting merged pages, namely, the merge process is contention-free.
\end{theorem}

\begin{proof}
As part of ensuring contention-free merge, we have already shown that merge operates 
on stable data (proven by Lemma~\ref{lem_stable}) and that there is no 
information loss as a result of the merge process (proven by Lemma~\ref{lem_safe}). 
Next we prove that the write path of the merge process does not overlap with the 
write path of users' transactions (i.e., writers). Recall that writers append new 
uncommitted tail records to tail pages (but as stated before uncommitted records 
do not participate in the merge), and writers perform in-place update of the 
\textit{Indirection} column within base records to point to the latest version 
of the updated records in tail pages (but the \textit{Indirection} column is not 
modified by the merge process), whereas the write path of the merge process 
consists of creating only a new set of read-only merged pages and eventually 
discarding the outdated base pages safely. 

Therefore, we must show that safely discarding base pag\-es does not interfere with 
users' transactions. In particular, as explained in Lemma~\ref{lem_safe}, if the 
original values were not written to tail records at the time of the update, then 
during the merge process, we were forced to store them somewhere or encounter 
information loss. It is not even clear where would be the optimal location for 
storing the original values. A simple minded approach of just adding them to tail 
pages would have broken the linear order of changes to records such that the older 
values would have appeared after the newer values, and it would have interfered 
with the ongoing update transactions. But, more importantly, the need to store 
the old values at any location would have implied that during the merge process 
multiple coordinated actions were required to ensure consistency across modification  
to isolated locations; hence, breaking the contention-free property of the merge. 
Therefore, by storing the original updated values at the time of update, we 
trivially eliminate all the potential contention during the merge process in 
order to safely discarding outdated base pages.

As a result, users' transactions are completely decoupled from the merge process, 
and users' transactions and the merge process do not contend over base, tail, or  
merged pages. \qed
\end{proof}

\subsubsection{Merge Performance Analysis \& Optimization}
When analyzing the performance of our merge algorithm, we observe that in 
the worst case, data from all updated columns for a given update range 
is read and written back, but it is a cost that is amortized over many updates, 
a key strength of our lineage-based storage architecture. In general, if updates 
are spread over a range of records (even if skewed), then the data for the entire 
range has to be read and written; however, when updates are strictly localized, 
then additional optimization can be applied to further prune the set of records read.

However, it is important to note that \tp's objective has been to introduce a 
contention-free merge procedure without the interference with the concurrent 
transactions because contention is the most important deciding factor in the 
overall performance of the system especially as the size of the main memory 
continues to increase (arguably the entire transactional data can fit in the 
main memory today) and the storage-class memories (such as SSDs) replace 
the mechanical disks~\cite{Diaconu:2013,DBLP:journals/pvldb/SadoghiCBNR14}. 
Nevertheless, we point out that there are potential opportunities to further 
improve the merge process execution time by employing and studying more complex 
join algorithms for implementing the merge by operating directly on the compressed  
data without the need to decompress and recompress the data (a complementary 
direction that is the outside the scope of the current work). Nevertheless, in 
our evaluation (cf. Section \ref{sec:exp}), with a single asynchronous merge thread, 
we were able to cope with tens of concurrent writer threads, and we were able 
to process millions of updates per second when updating 40\% of the columns 
on average. 

Lastly, we would like to point out that there are also a number of potential opportunities 
to guide the merge process in order to further accelerate relaxed analytical queries 
(those queries that can tolerate slightly outdated snapshot) by implicitly constructing 
a slightly outdated but consistent snapshot of the data across the entire table 
during the merge. Currently, our proposed merge is already relaxed and brings base pages 
almost up-to-date in time. Now we suggest to further coordinate the merge such that every merge 
not only take a set of consecutive committed tail records, but also takes only those consecutive 
committed records before an agreed upon time $t_i$. Thus, after merging a range of 
records, we are ensured that only committed records before the time $t_i$ is processed.
Furthermore, we propose that every page also maintains its temporal lineage to remember 
the timestamp of the earliest committed records that have not been merged yet, 
where ideally its timestamp is after $t_i$. Any range of records that yet to be merged 
or has failed to bring base pages forward in time up to $t_i$ can be manually brought 
forward to $t_i$ as part of the normal query processing by consolidating tail pages. 
Periodically, the agreed upon merge time is advanced from $t_i$ to $t_{i+1}$, and 
all subsequent merges are adjusted accordingly. However, exploiting the temporal 
lineage to speed up relaxed analytical queries (almost for free) is outside the scope 
of the current work.\footnote{Similar optimization for constructing an almost 
up-to-date and consistent snapshots was first introduced in~\cite{mohan92}, but 
it required to drain all active queries before the out-of-date snapshot could be 
advanced in time or it required maintaining multiple almost up-to-date snapshots 
simultaneously. Unlike the approach in~\cite{mohan92}, our proposed merge algorithm 
combined with the temporal lineage eliminates all contention with the ongoing queries 
such that the relaxed snapshot can be brought forward in time lazily and 
asynchronously.}

\begin{table*}[t]
\begin{center}
\begin{tabular}{| c | c | c | c | c | c | c | c | c | c |}  \hline
\textit{RID}	& \textit{Indirection}	& \makecell{\textit{Schema} \\ \textit{Encoding}} 	& \makecell{\textit{Start} \\ \textit{Time}} & \makecell{\textit{Last} \\ \textit{Updated} \\ \textit{Time}}	& \textit{Key}	& \textit{A}	& \textit{B}	& \textit{C} \\ [0.5ex] \hline\hline
\multicolumn{9}{|l|}{Recently merged records for the key range of $k_1$ to $k_3$; TPS = $t_7$}  \\ \hline
$b_1$ 			& $t_8$					& 0000$\ \ $	& 10:02 & 10:02			& $k_1$	& $a_1$	& $b_1$	& $c_1$   \\ \hline
$b_2$ 			& $t_{12}$					&{\color{red}0101$\ \ $}& 13:04&{\color{red}19:25}& $k_2$ &{\color{red}$a_{22}$}& $b_2$	&{\color{red}$c_{21}$}  \\ \hline
$b_3$ 			& $t_{11}$					& {\color{red} 0001$\ \ $}& 15:05& {\color{red}19:45} & $k_3$	& $a_3$	& $b_3$	& {\color{red}$c_{31}$}   \\ \hline
\multicolumn{9}{|l|}{Partitioned tail records for the key range of $k_1$ to $k_3$}  \\ \hline
$t_1$ 			& $b_2$					& 0100*				& 13:04 &		& $\varnothing$	& $a_2$			& $\varnothing$	& $\varnothing$   \\ \hline
$t_2$ 			& $t_1$					& 0100$\ \ $		& 19:21	&		& $\varnothing$	& $a_{21}$		& $\varnothing$	& $\varnothing$   \\ \hline
$t_3$ 			& $t_2$					& 0100$\ \ $		& 19:24	&		& $\varnothing$	& $a_{22}$		& $\varnothing$	& $\varnothing$   \\ \hline
$t_4$ 			& $t_3$					& 0001*				& 13:04	&		& $\varnothing$	& $\varnothing$	& $\varnothing$	& $c_2$   	\\ \hline
$t_5$			& $t_4$					& 0101$\ \ $		& 19:25	&		& $\varnothing$ & $a_{22}$		& $\varnothing$ & $c_{21}$  \\ \hline
$t_6$ 			& $b_3$					& 0001*				& 15:05	&		& $\varnothing$	& $\varnothing$	& $\varnothing$	& $c_3$   \\ \hline
$t_7$			& $t_6$					& 0001$\ \ $		& 19:45	&		& $\varnothing$	& $\varnothing$ & $\varnothing$	& $c_{31}$   \\ \hline
{\color{red} $t_8$}&{\color{red} $b_1$}	&{\color{red} 0000$\ \ $}&{\color{red}20:15}&& {\color{red}$\varnothing$}&{\color{red}$\varnothing$}&{\color{red}$\varnothing$}	&{\color{red}$\varnothing$}   \\ \hline
$t_9$			& $t_5$					& 0010*				& 13:04	&		& $\varnothing$	& $\varnothing$ & $b_2$	& $\varnothing$   \\ \hline
$t_{10}$		& $t_9$					& 0010$\ \ $		& 21:25	&		& $\varnothing$ & $\varnothing$	& $b_{21}$ & $\varnothing$  \\ \hline
{\color{red}$t_{11}$}& {\color{red}$t_7$}& {\color{red}0001$\ \ $}& {\color{red}21:30}&	& {\color{red}$\varnothing$}& {\color{red}$\varnothing$} & {\color{red}$\varnothing$}& {\color{red}$c_{32}$}   \\ \hline
{\color{red}$t_{12}$}& {\color{red}$t_{10}$}& {\color{red}0110$\ \ $}& {\color{red}21:55}&& {\color{red}$\varnothing$} & {\color{red}$a_{23}$} & {\color{red}$b_{21}$} & {\color{red}$\varnothing$}  \\ \hline
\end{tabular}
\caption{An example of the indirection interpretation and lineage tracking (conceptual tabular representation).}
\label{tab:example_lineage}
\end{center}
\end{table*}

\subsection{Maintaining In-Page Lineage}
\label{sec_lineage}
The lineage of each base page (and consequently merged pages) is maintained 
within each page independently as a result of the merge process. The in-page lineage information 
is instrumental to decouple the merge and update processing and to allow independent merging of the 
different columns of the same record at different points in time.  The in-page lineage information 
is captured using a rather simple and elegant concept, which we refer to as \textit{tail-page 
sequence number (TPS)} in order to keep track of how many updated entries (i.e., tail 
records) from tail pages have been applied to their corresponding base pages after 
a completion of a merge. Original base pages always start with TPS set to 0, a 
value that is monotonically increasing after every merge. Again to ensure this 
monotonicty property, we stressed earlier that always a consecutive set of committed 
tail records are used in the merge process.

TPS is also used to interpret the indirection pointer (also a monotonically 
increasing value) by readers after the merge is taken place. Consider our 
running example in Table~\ref{tab:example_merge}. After the first merge 
process, the newly merged pages have TPS set to 7, which implies that the first 
seven updates (tail records $t_1$ to $t_7$) in the tail pages have been applied 
to the merged pages. Consider the record $b_2$ in the base pages that has an 
indirection value pointing to $t_5$ (cf. Table~\ref{tab:example_merge}), 
there are two possible interpretations. If the transaction is reading the base 
pages with TPS set to 0, then the 5$^{th}$ update has not yet reflected on the 
base page. Otherwise if the transaction is reading the base pages with TPS 7, then 
the update referenced by indirection value $t_5$ has already been applied to 
the base pages as seen in Table~\ref{tab:example_merge}. Notably, the 
\textit{Indirection} column is updated only in-place (also a monotonically 
increasing value) by writers, while merging tail pages does not affect the 
indirection value. 

\begin{table*}[t]
\hspace{-3mm}
\begin{center}
\begin{tabular}{| c | c | c | c | c | c | c |}  \hline
\textit{RID}	& \textit{Indirection}	& \makecell{\textit{Schema} \\ \textit{Encoding}} 	& \textit{Start Time} & \textit{A}	& \textit{C} \\ [0.5ex] \hline\hline
\multicolumn{6}{|l|}{Merged, committed tail records for the key range of $k_1$ to $k_3$}  \\ \hline
$t_1$ 			& $b_2$					& 0100*				& 13:04			& $a_2$			& $\varnothing$   \\ \hline
$t_2$ 			& $t_1$					& 0100$\ \ $		& 19:21			& $a_{21}$		& $\varnothing$   \\ \hline
$t_3$ 			& $t_2$					& 0100$\ \ $		& 19:24			& $a_{22}$		& $\varnothing$   \\ \hline
$t_4$ 			& $t_3$					& 0001*				& 13:04			& $\varnothing$	& $c_2$   	\\ \hline
$t_5$			& $t_4$					& 0101$\ \ $		& 19:25			& $a_{22}$		& $c_{21}$  \\ \hline
$t_6$ 			& $b_3$					& 0001*				& 15:05			& $\varnothing$	& $c_3$   \\ \hline
$t_7$			& $t_6$					& 0001$\ \ $		& 19:45			& $\varnothing$ & $c_{31}$   \\ \hline
\multicolumn{6}{|l|}{Ordered, Inlined, Compressed committed tail records for the key range of $k_1$ to $k_3$}  \\ \hline
$c_1$ 			& $b_2$					& 0101$\ $		& {\color{red}13:04,19:21,19:24,19:25}&{\color{red}$a_2$,$a_{21}$,$a_{22}$,-}&{\color{red}$c_2$,-,-,$c_{21}$}  \\ \hline
$c_2$ 			& $b_3$					& 0001$\ $		& {\color{red}15:05,19:45}		&$\varnothing$& {\color{red}$c_3$,$c_{31}$} \\ \hline
\end{tabular}
\caption{An example of compressing merged tail pages (conceptual tabular representation).}
\label{tab:example_compress}
\end{center}
\end{table*}

More importantly, we can leverage the TPS concept to ensure read consistency 
of users' transactions when the mer\-ge is performed lazily and independently 
for the different col\-umns of the same records. Therefore, when the merge of 
columns is decoupled, each merge occurs independently and at different points  
in time. Consequently, not all base pages are brought forward in time 
simultaneously. Additionally, even if the merge occurs for all columns 
simultaneously, it is still possible that a reader reads base pages for 
the column \textit{A} before the merge (or during the merge before the page 
directory is updated) while the same reader reads the column \textit{C} 
after the merge; thus, reading a set of inconsistent base and merged pages.

\begin{lemma}
An inconsistent read with concurrent merge is always detectable.
\label{lem_read}
\end{lemma}

\begin{proof}
Since each base page independently tracks its lineage, i.e., its TPS counter;
therefore, TPS can be used to verify the read consistency. In particular, for 
a range of records, all read base pages must have an identical TPS counter; 
otherwise, the read will be inconsistent. Hence, an inconsistent read across 
different columns of the same record is always detectable. \qed
\end{proof}

\begin{theorem}
Constructing consistent snapshots with concurrent merge is always possible.
\end{theorem}

\begin{proof}
As proved in Lemma~\ref{lem_read}, the read inconsistency is always detectable. 
Furthermore, once a read inconsistency is encountered, then each page is simply 
brought to the desired query snapshot independently by examining its TPS and 
the indirection value and consulting the corresponding tail pages using the 
logic outlined earlier. Hence, consistent reads by constructing consistent 
snapshots across different columns of the same record is always possible. \qed
\end{proof}

TPS, or an alternative but similar counter conceptually, could be used as a 
high-water mark for resetting the cumulative updates as well. Continuing with 
our running scenario, in which we have the original base pages with the TPS 0 
(as shown in Table~\ref{tab:example_merge}), the merged pages the with TPS 
7 (as shown in Table~\ref{tab:example_lineage}). For simplicity, we assume 
the cumulation was also reset after the 7$^{th}$ tail record. For the record $b_2$, 
we see that the indirection pointer is $t_{12}$, for which we know that the 
cumulative update has been reset after the 7$^{th}$ update. This means that the tail 
record $t_{12}$ does not carry updates that were accumulated between tail records 
1 to 7. Suppose that the record was updated four times, where the update entries in the 
tail pages are $3^{rd}$, $5^{th}$, $10^{th}$, and $12^{th}$ tail records. The tail 
record $t_{5}$ is a cumulative and carries the updated values from the tail record $t_{3}$. 
However, the tail record $t_{10}$ is not cumulative (reset occurred at the 8$^{th}$ update), 
whereas the tail record $t_{12}$ is cumulative, but carries updates only from the tail 
record $t_{10}$ and not from $t_{5}$ and $t_{3}$. Suppose that a transaction is reading the 
base pages with the TPS 0, then to reconstruct the full version of the record $b_2$, it 
must read both the tail records $t_{5}$ and $t_{12}$ (while skipping 3$^{rd}$ and 
10$^{th}$). But if a transaction is reading from the merged pages with the TPS 7, 
then it is sufficient to only read the tail record $t_{12}$ to fully reconstruct 
the record because the 3$^{rd}$ and 5$^{th}$ updates have already been applied to 
the merged pages.



\subsection{Compressing Historic Data}
\label{sec_compress}
For historic tail pages, namely, the committed and subsequently merged tail 
pages, we introduce a contention-free compression sch\-eme to substantially 
reduce storage footprint and improving access patterns for historic queries 
(or time travel queries). Periodically, for a range of records, we compress 
only a set of merged tail pages (across all updated columns) that fall outside 
the oldest query snapshot in order to avoid clashing with the readers/writers 
of non-historic data. 

The key benefit of our compression scheme, which takes as input a set of 
tail pages for all updated columns for a given range of records (stable data) 
and outputs a set of newly compressed tail pages (in columnar form), are the 
following key properties (as demonstrated in Table~\ref{tab:example_compress}). 
First, the compressed tail records are re-ordered according to the base RID 
order; hence improving the locality of access. Second, for each record, and 
within each column, the different versions are stored inline and contiguously. 
The version inlining avoid the need to repeatedly store unchanged values due to 
cumulative updates, but, more importantly, it enables delta compression among 
the different versions of the records to further reduce the space overhead. Also 
collapsing the different versions of the same record into a single tail 
record eliminates the need for back pointers that are needed for referencing 
the previous versions. Thus, the inline versions are tightly packed and ordered 
temporally as shown in Table~\ref{tab:example_compress}.

The compressed tail pages are read-only and are used exclusively for historic 
queries. These new pages can be isolated and pushed down lower in the storage 
hierarchy since they are by definition colder and not accessed as frequently. 
Lastly, the page directory is updated by swapping the pointers for the old 
tail pages to point to the newly created compressed tail pages. Notably, 
considering that the access frequency is much lower for historic tail pages 
compared to the access frequency of non-historic tail pages, any locking- 
or non-locking-based approaches (such as the epoch-based approach discussed 
in Section~\ref{sec_merge}) can be employed without noticeably affecting the 
overall system performance.

\subsection{Record Partitioning Trade-offs}
\label{sec_range}
When choosing the range of records for partitioning (i.e., update range) there 
are several dimensions that needs to be examined. An important observation is that 
regardless of the range size, recent updates to tail pages 
will be memory resident and no random disk I/O is required; this trend 
is supported by continued increase in the size of main memory and the 
fact that the entire OLTP database is expected to fit in the main 
memory~\cite{Diaconu:2013,Larson:2015}. 

In our evaluation, we did an in-depth study of the impact of the range size, and 
we observed that the key deciding factor is the frequency at which the  
merges are processed. How frequent a merge is initiated is proportional 
to how many tail records are accumulated before the merge process is triggered. 
We further experimentally observed that the update range sizes in the order of $2^{12}$ to 
$2^{16}$ exhibit a superior overall performance vs. data fragmentation depending 
on the workload update distribution. Because for a smaller update range size, we 
may have many corresponding half-filled tail pages, but as the range size increases, the cost 
of half-filled tail pages are amortized over a much larger set of records.\footnote{To reduce 
space under-utilization, tail pages could be smaller than base pages, for instance, 
tail pages could be 4 KB while base pages are 32 KB or larger.}
Furthermore, the range size affects the clustering of updates in tail pages, the 
larger the range size, then it is more likely that cache misses occur when 
scanning the recent update that are not merged yet. Again, considering that 
recent cache sizes are in order of tens of megabytes, the choice of any range 
value between $2^{12}$ to $2^{16}$ is further supported. As noted before, 
one may choose a finer range partitioning for handling updates (i.e., update range), 
e.g., $2^{12}$, to improve locality of access while choosing coarser virtual range 
sizes when performing merges, essentially forcing the merge to take-in as input a set 
of consecutive update ranges that have been updated, e.g., choosing $2^4$ 
consecutive $2^{12}$ ranges in order to merge $2^{12} \times 2^4 = 2^{16}$ records.

For example, suppose the scan operation (even if there are concurrent 
scans) may access 2 columns, assume each column is $2^3$ bytes 
long. We further assume that the merge can keep up, namely, even for 
$2^{16}$ update range size, the number of tail records yet to be merged is less 
than $2^{16}$ (as shown in Section \ref{sec:exp}, such merging rate can be 
achieved while executing up to 16 concurrent update transactions). The 
overall scan footprint (combining both base pages and tail pages) is 
approximately $2^{16} \times 2^3 \times 2 \times 2 = 2^{21}$ (2 MB), 
which certainly fits in today's processor cache (in our evaluation, we 
used Intel Xeon E5-2430 processor, which has 15 MB cache size). Thus, 
even as scanning base records, if one is forced to perform random lookup 
within a range of $2^{16}$ tail records, the number of cache misses are 
limited compared to when the range size was beyond the cache capacity.

Another criteria for selecting an effective update range size is the need for 
RID allocation. In \tp, upon the first update to a range of records (e.g., $2^{12}$ 
to $2^{16}$ range), we pre-allocate $2^{12}$ to $2^{16}$ unused RIDs for referencing its 
corresponding tail pages. Tail RIDs are special in a sense they are not added to 
indexes and no unique constraint is applied on them. Once the tail RID range is fully 
used, then either a new unused RID range is allocated or an existing underutilized 
tail RID range can be re-assigned (partially used RID range must satisfy TPS 
monotonicity requirement). Furthermore, in order to avoid overlapping the base 
and tail RIDs, one could assign tail RIDs in the reverse order starting from $2^{64}$; 
therefore, tail RIDs will be monotonically decreasing, and the TPS logic must be 
reversed accordingly. The benefit of reverse assignment is that while scanning page 
directory for base pages, there is no need to first read and later skip tail page 
entries (read optimization).

\section{Fast Transactional Capabilities}
\label{sec_rec}
In order to support concurrent transactions (where each transaction may 
consists of many statements), any database engine must provide necessary 
functionalities to ensure the correctness of concurrent reads and writes 
of the shared data. Furthermore, transaction logging is required in order to 
recover the system from crash and media failure. 

\subsection{Simplified Concurrency \& Recovery}
We begin by first briefly highlighting the key features of our concurrency control without the need 
of locking followed by the discussion of low-level latching and logging 
requirements. 


\subsubsection{Optimistic Concurrency Control}
\label{sec_transaction}
It is important to note that \tp's focus is on storage architecture thus making 
it agnostic to the underlying concurrency protocol. In our \tp\ prototype, we relied on 
the recently proposed optimistic concurrency model introduced in~\cite{DBLP:journals/pvldb/SadoghiCBNR14} 
that supports full ACID properties for multi-statement transactions, and we also employed the 
speculative reads proposed in~\cite{LarsonBDFPZ11}.

When a transaction starts, it receives a begin time from a synchronized clock (time is advanced 
before it is returned) and is assigned a unique monotonically increasing transaction ID.\footnote{In 
fact, the begin time could itself be used as a seed for the transaction ID.} The begin time is 
used to determine records visibility, namely, only the latest version of records that were 
created/modified before the begin time are visible to the transaction. Similarity before 
transaction commits, it receives a commit time (again time is advanced before it is returned). 
All modifications by the transaction (including insertion and updates) will use the commit time 
as the start time of the inserted/updated records. Recall that the start time of the record 
acts as an implicit end time of the previous version of the record. The transaction manager 
also maintains the state of each transaction and its begin/commit time in a hashtable. Each 
transaction has four states: active, pre-commit, committed, and aborted. When a transaction begins 
it starts in ``active'' state and once completes all its reads and writes operations, it enters 
``pre-commit'' state to validates its reads before reaching the ``commits'' state. A failed 
transaction will have an ``abort'' state. 

The employed optimistic concurrency control is formalized as follows:\footnote{In the interest of 
space, we are presenting a much simpler variation of the concurrency models that were presented 
in~\cite{LarsonBDFPZ11,DBLP:journals/pvldb/SadoghiCBNR14}, for complete details of the protocol 
please refer to the original papers. In particular, the following protocol can incorporate the 
idea of certify protocol using latch-free read counters proposed in~\cite{DBLP:journals/pvldb/SadoghiCBNR14} 
in order to reduce the likelihood of validation failure.}

\begin{itemize}
\item \textit{read} $r(x)$, the latest committed and visible version of $x$ is read. The latest version 
is read by accessing the base record and examining its \textit{Indirection} column. If the indirection 
is set to null ($\bot$) or the indirection value is not larger than the TPS counter of the base page, then 
the base record holds the latest version of the record; otherwise, the latest version of the record is fetched by 
following the indirection forward pointer. If the latest version of the record is visible, i.e., meaning 
record's start time is smaller than transaction begin time, then the latest version of the record is returned; 
otherwise, the last committed and visible version of the record is returned by following the indirection backward pointer. 
The \textit{Start Time} column may also hold transaction ID, and in such cases the transaction manager is consulted 
to determine whether the record is committed and visible. The RID of the committed and visible record 
is added to the transaction's readset.

\item \textit{speculative-read} $r(x)$, the latest pre-committed and visible version of $x$ is read. 
The speculative read is similar to normal read except that it relaxes the read criteria, and it also 
allows reading updated/inserted records by those transactions that are in the pre-commit state  
and not committed yet. The RID of the pre-committed and visible record is added to the transaction's 
readset.

\item \textit{write} $w(x)$, a new uncommitted version of record is written, a write-write conflict 
is detected prior to modifying $x$. If no conflict is detected, a new uncommitted version of $x$ is 
written and the indirection value of the base record is updated accordingly. Detecting write-write 
conflict is a two-step process that exploits the indirection value of the base record. Each indirection 
pointer reserves one bit for latching. First, the latch bit of the indirection value is set using 
atomic compare-and-swap (CAS) operator. If setting the latch bit fails, then it is an indicator of 
write-write conflict, and the transaction aborts. Second, once the latch bit is set, then the 
start time of the latest version of the record is checked. If the start time holds a transaction 
ID of a competing uncommitted transaction, then another write-write conflict is detected, the latch 
bit is released and the transaction aborts. If no write-write conflict is detected, then the new 
version of the record is installed in the tail page, and the tail RID is written in the \textit{Indirection} 
column of the base record. The start time of the newly installed record is set to the transaction ID 
of the writer. The latch bit of the record is released.

\item \textit{validate reads}, for each item $x$ in the readset, the read repeatability is verified. 
First, a commit timestamp is acquired for the transaction and the transaction state is changed from 
``active'' to ``pre-commit''; both changes are reflected atomically in the transaction manager's hash\-table. 
For each read record, if the currently committed and visible RID based on the commit time of the 
transaction is equal to the committed (or pre-committed for speculative reads) and visible RID as of 
the begin time of the transaction,\footnote{Committed/pre-committed visible RIDs as of the begin 
time of the transaction are already maintained in the readset, so no additional computation is required.} 
then the validation is satisfied; otherwise, the validation fails and the transaction is aborted and 
rolled back.

\item \textit{commit}, to finalize the transaction, the necessary log rec\-ords are written and transaction 
changes are made visible to all other transactions. When the transaction is committed, the transaction 
manager atomically changes the state of transaction from ``pre-commit'' to ``committed''. Note, 
the newly written records will have transaction ID in their \textit{Start Time} column, and we do not swap the 
transaction ID with the commit time as part of the commit phase because, in general, commit protocol 
is extremely time-sensitive, and it is essential to minimize the length of commit phase to avoid overall 
performance deterioration. Swapping the transaction ID with commit time is done lazily by future 
readers.
\end{itemize}

The validation in the optimistic concurrency is only need\-ed for repeatable read and serializability. 
The read committed isolation always reads the visible and committed version and does not 
require validation, and the snapshot isolation reads the view of the database from an instantaneous 
point in time, which again does not require validation except for any records that were read 
speculatively.

\comment{
\subsection{Optimistic Concurrency Control Protocol}
In terms of concurrency protocol for transaction processing, any existing protocols 
can be leveraged because \tp's primary focus is the storage architecture. 
In particular, we relied on the recently proposed optimistic concurrency model 
introduced in~\cite{DBLP:journals/pvldb/SadoghiCBNR14} that supports full ACID properties 
for multi-statement transactions, and we also employed the speculative reads proposed 
in~\cite{LarsonBDFPZ11}.

When a transaction starts, it receives a begin time from a synchronized clock (time is advanced 
before it is returned) and is assigned a unique monotonically increasing transaction ID.\footnote{In 
fact, the begin time could itself be used as a seed for the transaction ID.} The begin time is 
used to determine records visibility, namely, only the latest version of records that were 
created/modified before the begin time are visible to the transaction. Similarity before 
transaction commits, it receives a commit time (again time is advanced before it is returned). 
All modifications by the transaction (including insertion and updates) will use the commit time 
as the start time of the inserted/updated records. Recall that the start time of the record 
acts as an implicit end time of the previous version of the record. The transaction manager 
also maintains the state of each transaction and its begin/commit time in a hashtable. Each 
transaction has four states: active, pre-commit, committed, and aborted. When a transaction begins 
it starts in ``active'' state and once completes all its reads and writes operations, it enters 
``pre-commit'' state to validates its reads before reaching the ``commits'' state. A failed 
transaction will have an ``abort'' state. 

The employed optimistic concurrency control is formalized as follows:\footnote{In the interest of 
space, we are presenting a much simpler variation of the concurrency models that were presented 
in~\cite{LarsonBDFPZ11,DBLP:journals/pvldb/SadoghiCBNR14}, for complete details of the protocol 
please refer to the original papers. In particular, the following protocol can incorporate the 
idea of certify protocol using latch-free read counters proposed in~\cite{DBLP:journals/pvldb/SadoghiCBNR14} 
in order to reduce the likelihood of validation failure.}

\begin{itemize}[leftmargin=*,itemsep=0pt]
\item \textit{read} $r(x)$, the latest committed and visible version of $x$ is read. The latest version 
is read by accessing the base record and examining its \textit{Indirection} column. If the indirection 
is set to null ($\bot$) or the indirection value is not larger the base page TPS counter, then the base 
record holds the latest version of the record; otherwise, the latest version of the record is fetched by 
following the indirection forward pointer. If the latest version of the record is visible, i.e., meaning 
record's start time is smaller than transaction begin time, then the record is returned; otherwise the 
last committed and visible version of the record is returned by following the indirection backward pointer. 
The start time column may also hold transaction ID, and in such cases the transaction manager is consulted 
to determine whether the record is committed and visible. The RID of the committed and visible record 
is added to the transaction's readset.

\item \textit{speculative-read} $r(x)$, the latest pre-committed and visible version of $x$ is read. 
The speculative read is similar to normal read except that it relaxes the read criteria, and it also 
allows reading updated/inserted records by those transactions that are in the pre-commit state  
and not committed yet. The RID of the pre-committed and visible record is added to the transaction's 
readset.

\item \textit{write} $w(x)$, a new uncommitted version of record is written, a write-write conflict 
is detected prior to modifying $x$. If no conflict is detected, a new uncommitted version of $x$ is 
written and the indirection value of the base record is updated accordingly. Detecting write-write 
conflict is a two-step process that exploits the indirection value of the base record. Each indirection 
pointer reserves one bit for latching. First, the latch bit of the indirection value is set using 
atomic compare-and-swap (CAS) operator. If setting the latch bit fails, then it is an indicator of 
write-write conflict, and the transaction aborts. Second, once the latch bit is set, then the 
start time of the latest version of the record is checked. If the start time holds a transaction 
ID of a competing uncommitted transaction, then another write-write conflict is detected, the latch 
bit is released and the transaction aborts. If no write-write conflict is detected, then the new 
version of the record is installed in the tail page, and the tail RID is written in the \textit{Indirection} 
column of the base record. The start time of the newly installed record is set to the transaction ID 
of the writer. The latch bit of the record is released.

\item \textit{validate reads}, for each item $x$ in the readset, the read repeatability is verified. 
First, a commit timestamp is acquired for the transaction and the transaction state is changed from 
``active'' to ``pre-commit''; both changes are reflected atomically in the transaction manager's hashtable. 
For each read record, if the currently committed and visible RID based on the commit time of the 
transaction is equal to the committed (or pre-committed for speculative reads) and visible RID as of 
the begin time of the transaction,\footnote{committed/pre-committed visible RIDs as of the begin 
time of the transaction are already maintained in the readset, so no additional computation is required.} 
then the validation is satisfied; otherwise, the validation fails and the transaction is aborted and 
rolled back.

\item \textit{commit}, to finalize the transaction, the necessary log records are written and transaction 
changes are made visible to all other transactions. When the transaction is committed, the transaction 
manager atomically changes the state of transaction from ``pre-commit'' to ``committed''. Note, 
the newly written records will have transaction ID in their start time, and we do not swap the 
transaction ID with the commit time as part of the commit phase because, in general, commit protocol 
is extremely time-sensitive, and it is essential to minimize the length of commit phase to avoid overall 
performance deterioration. Swapping the transaction ID with commit time is done lazily by future 
readers.
\end{itemize}

The validation in the optimistic concurrency is only needed for repeatable read and serializability. 
The read committed isolation always reads the visible and committed version and does not 
require validation, and the snapshot isolation reads the view of the database from an instantaneous 
point in time, which again does not require validation except for any records that were read 
speculatively. 
}

\subsubsection{Low-level Synchronization Protocol}
In terms of low-level latching, our lineage-based storage has a set of unique benefits, namely, 
readers do not have to latch the read-only base pages or fully committed tail pages. Also there is 
no need to latch partially committed tail pages when accessing committed records. More importantly, 
writers never modify base pages (except the \textit{Indirection} column) nor the fully committed tail 
pages, so no latching is required for stable pages. The \textit{Indirection} column is at most 8-byte
long; therefore, writers can simply rely on atomic compare-and-swap (CAS) operators  
to avoid latching the page. 

As part of the merge process, no latching of tail 
and base pages are required because they are not modified. The only latching requirement 
for the merge is updating the page directory to point to the newly created merged pages.
Therefore, every affected page in the page directory are latched one at a time to 
perform the pointer swap or alternatively atomic CAS operator is employed for each 
entry (pointer swap) in the page directory. Alternatively, the page directory 
can be implemented using latch-free index structures such as Bw-Tree~\cite{bwtree}.

\subsubsection{Recovery and Logging Protocol}
Our lineage-based storage architecture consists of read-only base pages 
(that are not modified) and append-only updates to tail pages (which are not modified 
once written). When a record is updated, no logging is required for base pages 
(because they are read-only), but the modified tail pages requires redo 
logging. Again, since we eliminate any in-place update for tail pages, 
no undo log is required. Upon a crash, the redo log for tail pages are 
replayed, and for any uncommitted transactions (or partial rollback), 
the tail record is marked as invalid (e.g., tombstone), but the space 
is not reclaimed until the compression phase (cf. Section~\ref{sec_compress}). 

The one exception to above rule for logging and recovery is the 
\textit{Indirection} column, which is updated in-place. There are two possible recovery 
options: (1) one can rely on standard undo-redo log for the \textit{Indirection} 
column only or (2) one can simply rebuild the \textit{Indirection} column upon crash. 
The former option can further be optimized based on the realization that tail pages 
undergo strictly redo policy and aborted transactions do not physically remove the 
aborted tail records as they are only marked as tombstones. Therefore, it is acceptable 
for the \textit{Indirection} column to continue pointing to tombstones, and from 
the tombstones finding the latest committed values. As a result, even for the 
\textit{Indirection} column only the redo log is necessary. For the latter recovery 
option, as discussed earlier, to speedup the merge process, we materialize the 
\textit{Base RID} column in tail records that can be used to populate the 
\textit{Indirection} column after the crash. Alternatively, even without 
materializing an additional RID column, one can follow backpointers in the 
\textit{Indirection} column of tail records to fetch the base RID because 
the very first tail record always points back to the original base record.

The merge process is idempotent because it operates str\-ictly on committed data and 
repeated execution of the merge always produce the exact same results given a 
set of base pages, their corresponding tail pages, and a merge threshold that dictates 
how many consecutive committed tail records to be used in the merge process. Therefore, 
only operational logging is required for the merge process. Also updating the entries in 
the page directory upon completion of the merge process simply requires standard index 
logging (both undo-redo logs). If crash occurs during the merge, simply the partial merge 
results can be ignored and the merge can be restart\-ed. Similarly, compressing the 
historic tail-pages is idempotent and requires only operational logging and restart 
on crash (cf. Section~\ref{sec_compress}).

\subsection{Miscellaneous: Write-ahead Logging}
In the columnar storage, writing the log record becomes even a more expensive 
operation because a single record update or insert spans multiple columns that reside 
on different pages. As a result, when writing the log record, all affected 
pages due to updates must be latched with an exclusive access. While the pages 
are latched, the necessary changes are made to the pages, the log record is 
written, the log sequence number (LSN) is acquired, and the \textit{pageLSN} 
for every page is updated to hold the LSN of the latest update~\cite{ARIES}. 
Subsequently, all the exclusive latches are released. Holding an exclusive 
latch is essential otherwise the page may end up in an inconsistent state. 

Consider two transactions attempting to update two different rec\-ords on the same 
page. The first transaction $t_1$ receives LSN $l_1$ while the second 
transaction $t_2$ receives LSN $l_2$. Suppose log rec\-ords are written without 
holding the latch on the page that needs to be updated. In a concurrent system, 
it is possible that $t_2$ first gets to update the page, then updates the 
\textit{pageLSN} to $l_2$ while $t_1$ has yet to update the page. At this 
point, the page has entered an inconsistent state because the \textit{pageLSN} 
indicates that all updates up to $l_2$ are applied to the page while $t_1$ has 
not made any changes. Now if $t_1$ arrives and make the desired updates, then 
it is still unclear how $t_1$ can announce that its changes are applied because 
if it attempts to update the \textit{pageLSN} to $l_1$, then the page again will 
be in an inconsistent state because both updates from $t_1$ and $t_2$ have been 
applied, but the \textit{pageLSN} only indicates that update from $t_1$ has been 
applied. At this point, the page is dirty with an incorrect \textit{pageLSN} again. 
Suppose the bufferpool stealing policy is now exercised, and the dirty page 
with the \textit{pageLSN} $l_1$ (but having updates from both $t_1$ and $t_2$) 
is written to disk. Now consider the scenario, in which the page is flushed and 
the transaction $t_1$ commits, but while $t_2$ is still active the database 
crashes. Now after the crash, the page in question appears to be clean because 
the latest change on the page denoted by $l_1$ has committed; hence, dirty 
uncommitted changes from $t_2$ will not be discovered, and the database will be 
inconsistent. An alternative scenario is when the page is flushed to disk 
(again due to bufferpool stealing policy) after $t_2$ updates the \textit{pageLSN}, 
but before $t_1$ make any changes to the page. Thus, the inconsistent data becomes 
persistent. Suppose $t_1$ request the page again, and the page is brought back to 
memory, and $t_1$ updates the page content (and leaves the \textit{pageLSN} to remain 
at $l_2$), and $t_1$ commits. Now if the database crashes, the persistent page on disk has the 
\textit{pageLSN} $l_2$ and both $t_1$ and $t_2$ have committed, but changes from 
$t_2$ are only reflected on disk and updates from $t_1$ are lost. However, the 
crash recovery protocol sees that the page is up-to-date and clean and no 
further action is taken; hence, the database will remain in an inconsistent 
state.

Although holding an exclusive latch while updating a page solves the above mentioned 
inconsistencies, but holding an exclusive latch for such a prolonged duration will 
substantially deteriorates the overall performance of the system. In \tp, the 
base pages are not updated in-place, and tail pages follow a strict append-only 
policy. But even with an append-only policy if the write-ahead logging is employed, 
then it may deteriorate the transaction throughput due to exclusive latches. Similar 
challenges also arise even when updating a single column, namely, the \textit{Indirection} 
column (but without the need of undo log). To address the WAL challenges for 
columnar storage, we introduce an \ownfull\ protocol to substantially reduce 
the need for holding exclusive latches in order to correctly update the 
\textit{pageLSN}. The key idea behind \own\ protocol is to have all writers to 
hold a compatible shared latch instead (not exclusive latches) while only one 
transaction (with the highest LSN) is selected as the owner of the page and 
responsible for updating the \textit{pageLSN} and promoting its shared latch 
to an exclusive one. Other transactions do not need to update the \textit{pageLSN}; 
however, before releasing the shared latches, they must ensure that the page 
has an owner. Therefore, we propose that each data page to have an \textit{ownerLSN} 
in addition to the \textit{pageLSN}.\footnote{The \textit{ownerLSN} does not 
have to be materialized on the page itself and could be maintained as a 
meta-data in an external data structure.} 

For multi-threaded in-place update of the \textit{Indirection} column, our 
proposed \own\ protocol will be as follows. All writers (i.e., update 
transactions) acquire a shared latch on the \textit{Indirection} column before 
making any changes. After the latch is granted, the \textit{Indirection} column 
is updated in-place, the redo log record is written, and the LSN is 
acquired. If the \textit{ownerLSN} of the page is larger than the writer's LSN, 
then the shared latch is released. However, if the \textit{ownerLSN} is smaller  
than the writer's LSN, then the writer updates the \textit{ownerLSN} (using 
atomic CAS operator) and promotes its shared latch to an exclusive latch (a 
conditional promotion), and checks if it is still the owner while waiting 
otherwise the latch is released. Once the exclusive latch is granted, the writer 
will update the \textit{pageLSN} and release the latch. 

Therefore, if there are 100 concurrent writers, then only one writer will get 
an exclusive latch on be-half of all the writers, and the \textit{pageLSN} is 
updated once for every 100 concurrent writers. Since the writer will never release 
its shared latch as long as it is deemed owner, then we will never flush (nor persist) 
a page when the page content and the \textit{pageLSN} are not yet consistent. 
Also to ensure that the page is flushed eventually to avoid starvation, periodically 
a page is forced to drain all its current writers and update \textit{pageLSN} before 
accepting any new writers (similar to a checkpointing procedure). This can be 
implemented by ensuring that at most $\theta_s$ shared latches are granted between 
any two consecutive flushes. Therefore, once the threshold $\theta_s$ is exceeded, 
then no new shared latches are granted for writers and the page is forced to drain 
its writers and be flushed subsequently. Recall that readers do not need to hold 
any shared latch when reading the \textit{Indirection} column because all changes 
to the page are done using atomic CAS operator; thus, the forced flushing policy 
does not affect the readers.

For multi-threaded append-only updates (sparse inserts) or inserts to tail 
pages, we propose slightly more complicated variation of our \own\ protocol.
We further assume that every page is pre-allocated with a set of fixed size 
slots (the number of available slot is maintained for each page) to accommodate 
the updates. The writer first acquires a shared latch on the \textit{Indirection} 
column in the tail page followed by acquiring a tail RID for appending the updates 
(or a new record). For appending a tail record, each writer acquires a shared latch 
for each updated column individually. The latches are acquired in the order in 
which columns appear in the table schema (from left to right) and held until 
the update is completed. Once the shared latches are granted, the writer writes 
the new value in the pre-allocated slots determined by the assigned tail RID. After append 
is completed, the redo log record is written, and the LSN is acquired. For 
every page, if the writer's LSN is the highest LSN that the page has seen so 
far (LSN $\geq$ \textit{ownerLSN}), then the writer becomes the page owner 
and updates the \textit{ownerLSN} using atomic CAS operator; 
otherwise the shared latch is released. For each page that the writer has the 
ownership, the writer promotes its shared latch to an exclusive latch, and checks 
if it is still the owner while waiting for the latch; otherwise, it releases the 
latch. While holding the exclusive latch, the writer will update the \textit{pageLSN}, 
and optionally will compact slots and pre-allocate new slots for future updates as needed. 
Notably even if the writer with the ownership is aborted, the tail entries will not be 
removed; thus, the write will continue to update the \textit{pageLSN} accordingly.

Handling the starvation problem for page flushing is much simpler for tail pages, 
and it requires no intervention. As soon as a tail page is full, then naturally 
it will have no more writers, so it can be flushed without the need to introduce 
any forced flushing policy. 
\section{Experimental Evaluation}
\label{sec:exp}
In order to study the impact of high-throughput transaction processing in the presence of 
long-running analytical queries, we carried out a comprehensive set of experiments. These experiments were run using an 
existing micro benchmark proposed in~\cite{LarsonBDFPZ11,DBLP:journals/pvldb/SadoghiCBNR14}, 
for the sake of a fair comparison and evaluation. This benchmark allows us to study different 
storage architectures by narrowing down the impact of concurrency with respect to the 
database active set by adjusting the degree of contention between readers and writers. 

\subsection{Experimental Setting}
We evaluate the performance of various aspects of our real-time OLTP 
and OLAP system. Our experiments were conducted on a two-socket Intel Xeon E5-2430 
@ 2.20 GHz server that has 6 cores per socket with hyper-threading enabled 
(providing a total of 24 hardware threads). The system has 64 GB of memory 
and 15 MB of L3 cache per socket. We implemented a complete working prototype 
of \tp\ and compared it against two different techniques, (i) \iph\ and 
(ii) \dbt, which are described subsequently. The prototype was implemented 
in Java (using JDK 1.7).
Our primary focus here is to simultaneously evaluate read and write throughputs 
of these systems under various transactional workloads concurrently executed 
with long-running analytical queries, which is the key characteristic of any 
real-time OLTP and OLAP system.

Our employed micro benchmark defined in~\cite{LarsonBDFPZ11,DBLP:journals/pvldb/SadoghiCBNR14} consists of three key types 
of workloads: (1) low contention, where the database active set is 10M records; (2) medium 
contention, where the active set is 100K records; and (3) high contention, where the 
active set is 10K records. It is important to note that the database size is not 
limited to the active set and can be much larger (millions or billions of records). 
Similar to~\cite{LarsonBDFPZ11,DBLP:journals/pvldb/SadoghiCBNR14}, we consider two classes of transactions: 
\textit{read-only transactions} (executed under snapshot isolation semantics) that scan 
up to 10$\%$ of the data (to model TPC-H style analytical queries) and short update 
transactions (to model TPC-C and TPC-E transactions), in which each \textit{short update 
transaction} consists of 8 read and 2 write statements (executed under committed read semantics) 
over a table schema with 10 columns. In addition, we vary the ratio of read/writes in 
these update transactions to model different customer scenarios with different read/write 
degrees. By default, transactional throughput of these schemes are evaluated while 
running (at least) one scan thread and one merge thread to create the real-time OLTP 
and OLAP scenario. Unless stated explicitly, the percentage of reads and writes in the 
transactional workload is fixed at 80$\%$ and 20$\%$, respectively. On average 40\% 
of all columns are updated by the writers. Lastly, the page size is set to 32 KB for both 
base and tail pages because a larger page size often results in a higher compression 
ratio suitable for analytical workloads~\cite{Kemper:2011}.

\begin{figure}[t]
\begin{center}
\subfigure[Low Contention.]{\centering \includegraphics[trim=0cm 0cm 0cm 0cm, width=3in]{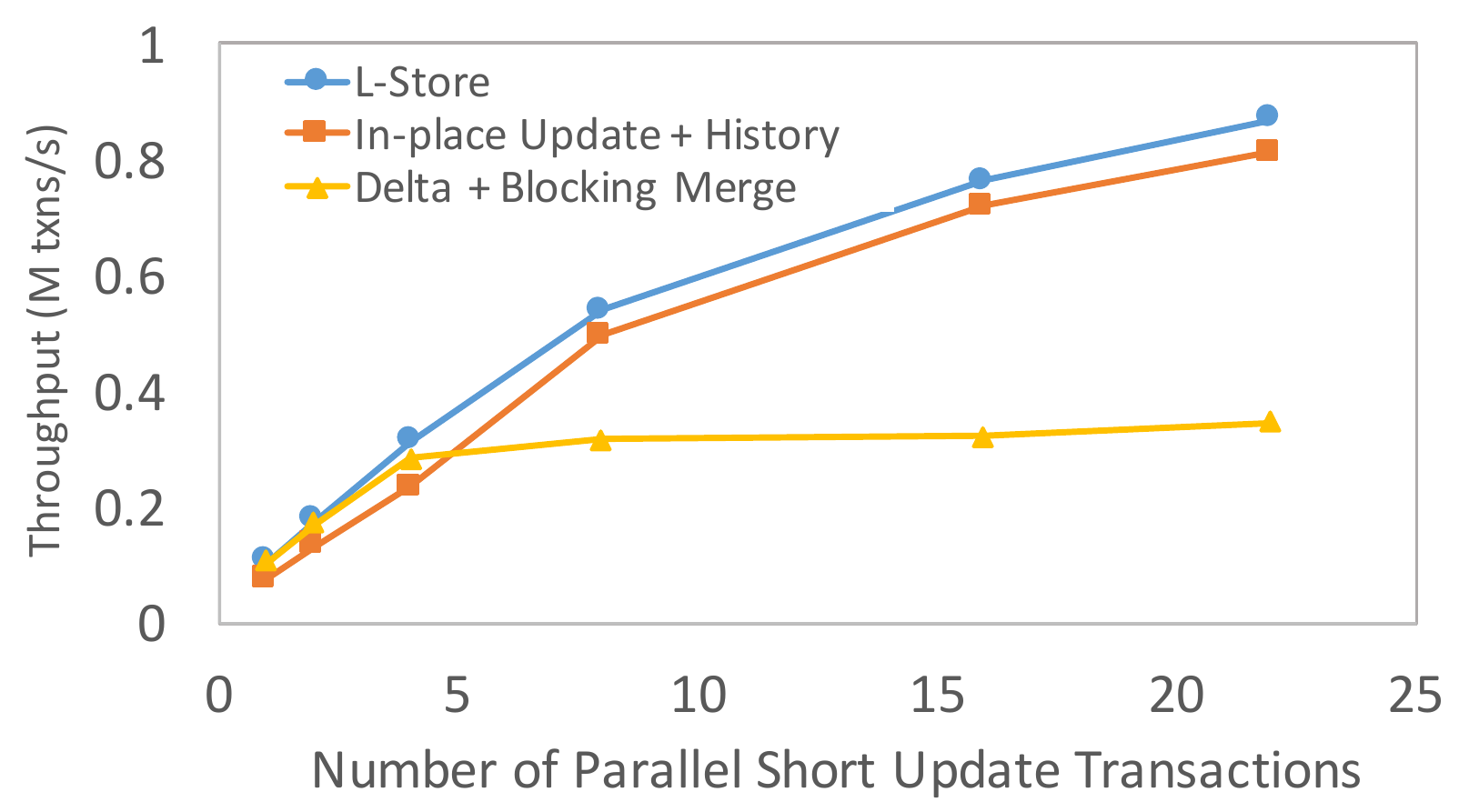} \label{fig:throughput-lc}}
\subfigure[Medium Contention]{\centering \includegraphics[trim=0cm 0cm 0cm 0cm, width=3in]{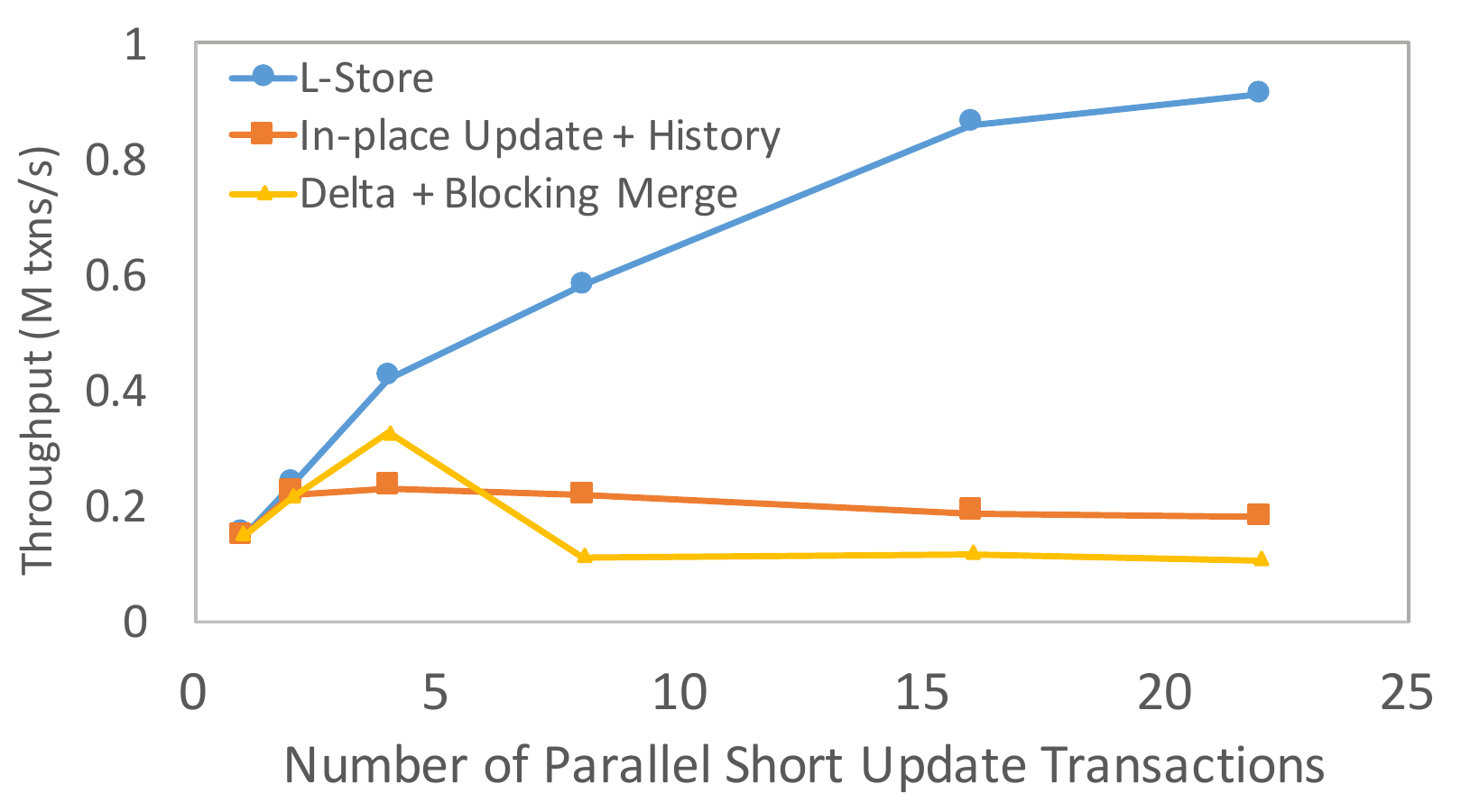} \label{fig:throughput-mc}}
\subfigure[High Contention]{\centering \includegraphics[trim=0cm 0cm 0cm 0cm, width=3in]{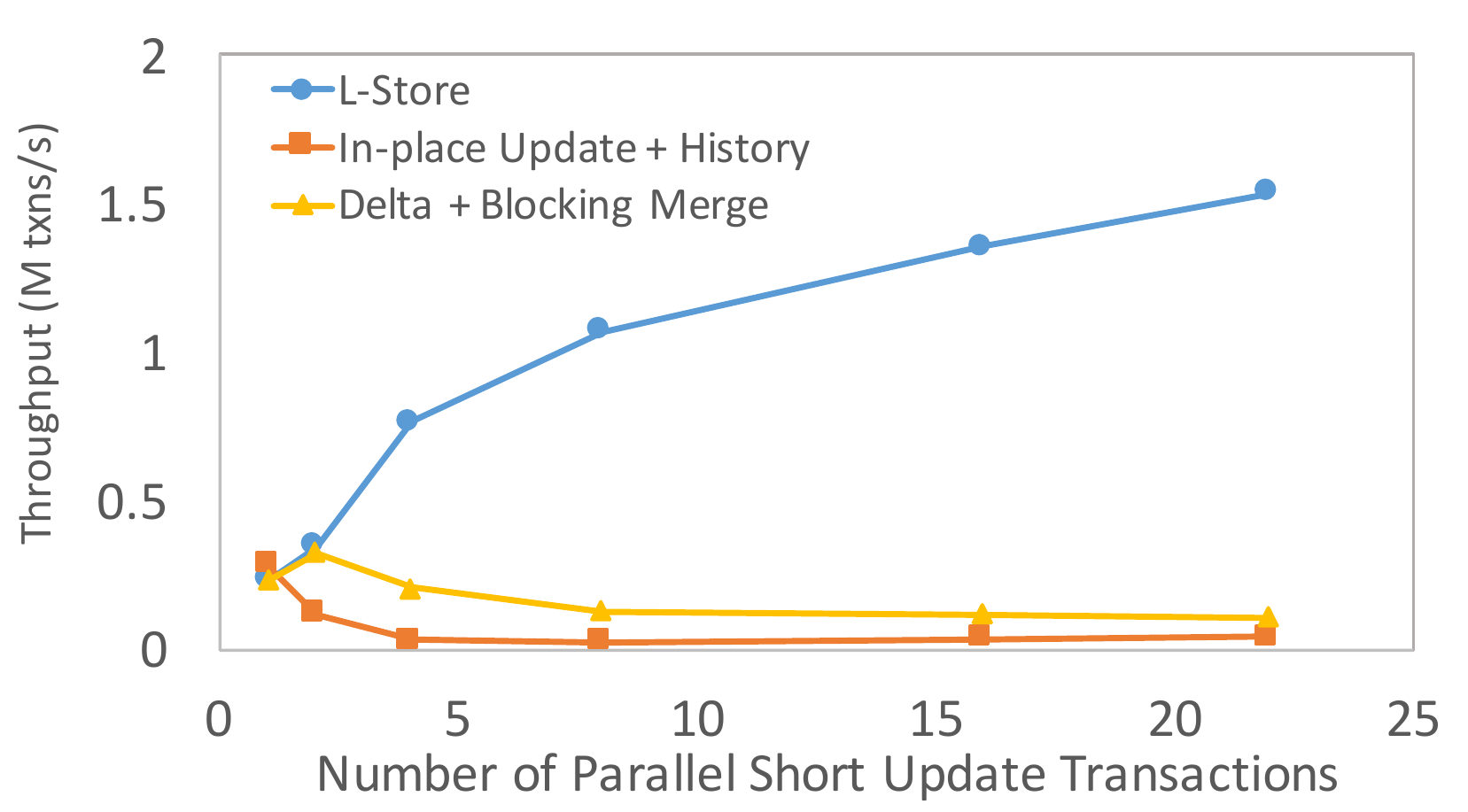} \label{fig:throughput-hc}}
\end{center}
\caption{Scalability under varying contention level.}
\label{fig:lrt-mc}
\end{figure}

Next we describe the two techniques that are compared with \tp. 
We point out the primary features of these techniques and describe it with respect to 
\tp. For fairness, across all techniques, we have maintained columnar storage, 
maintained a single primary index for fast point lookup, and employed the 
embedded-indirection column to efficiently access the older/newer versions of 
the records. Additionally, logging has been turned off for all systems as  
logging could easily become the main bottleneck (unless sophisticated logging 
mechanisms such as group commits and/or enterprise-grade SSDs are employed).
In the \iph\ technique, we are required to write both redo-undo logs for all updates 
while for \tp\ and \dbt\ only redo log is needed due to their append-only 
scheme.

\textbf{In-place Update + History (IUH):} 
A prominent storage organization is to append old versions of records 
to a history table and only retain the most recent version in the main table,
updating it in-place. An example of a commercial system that has implemented this 
table organization is the Oracle Flashback Archive~\cite{oracleArchive:webpage}; 
thus, our \iph\ is inspired by such table organization that avoids having multiple 
copies and representations of the data. However, due to the nature of the in-place update 
approach, each page requires standard shared and exclusive latches that are often found in major 
commercial database systems. In addition to the page latching requirement, if a transaction 
aborts, then the update to the page in the main table is undone, and the previous 
record is restored. Scans are performed by constructing a consistent snapshots, 
namely, if records in the main table are invisible with respect to query's read time, 
then the older versions of the records are fetched from the history table by 
following the indirection column. In our implementation of \iph, we also ignored 
other major costs of in-place update over the compressed data, in which the new value 
may not fit in-place due to compression and requires costly page splits or shifting 
data within the page as part of update transactions. We further optimized the 
history table to include only the updated columns as opposed to inserting all 
columns naively.

\textbf{Delta + Blocking Merge (DBM):} This technique is inspired by HANA~\cite{Krueger11}, 
where it consists of a main store and a delta store, and undergoes a periodic 
merging and consolidation of the main and delta stores. However, the periodic merging 
requires the draining of all active transactions before the merge begins and after the merge ends. 
Although the resulting contention of the merge appears to be limited to only the boundary 
of the merge for a short duration, the number of merges and the frequency 
at which this merge occurs has a substantial impact on the overall performance. 
We optimized the delta store implementation to be columnar and included 
only the updated columns~\cite{Plattner14}. Additionally, we applied our range 
partitioning scheme to the delta store by dedicating a separate delta store for 
each range of records to further reduce the cost of merge operation in presence 
of data skew. The partitioning allow us to avoid reading and writing the unchanged 
portion of the main store.

\subsection{Experimental Results}
In what follows, we present our comprehensive evaluation results in order to 
compare and study our proposed \tp\ with respect to state-of-the-art approaches.

\textbf{Scalability under contention:}
In this experiment, we show how transaction throughput scales as we increase 
the number of update transactions, in which each update transaction is assigned 
to one thread. For the scalability experiment, we fix the number of reads 
to 8 and writes to 2 for each transaction against a table with active set of $N = 10$ 
million rows. Figure~\ref{fig:throughput-lc} plots the transaction throughput (y-axis) 
and the number of update threads (x-axis). Under low contention, the throughput for \tp\ and 
\iph\ scales almost linearly before data is spread across the two NUMA nodes. 
The \dbt\ approach however does not scale beyond a small number of threads due to the 
draining of active transaction before/after of each merge process, which brings down 
the transaction throughput noticeably. With increasing number of threads, the number 
of merges and the draining of active transactions become more frequent, which reduces the 
transaction throughput significantly. The \iph\ approach has lower throughput compared 
to \tp\ due to the exclusive latches held for data pages that block the readers 
attempting to read from the same pages. The presence of a single history table also 
results in reduced locality for reads and more cache misses.

In addition, we study impact of increasing the degree of contention by varying the size of the 
database active set. For a fixed degree of contention, we vary the number of parallel 
update transactions from 1 to 22. For both medium contention (Figure~\ref{fig:throughput-mc}) 
and high contention (Figure~\ref{fig:throughput-hc}), we observe that \tp\ consistently 
outperforms the \iph\ and \dbt\ techniques as the number of parallel transactions 
is increased. For medium contention, we observed a speedup of up to 5.09$\times$ compared 
to the \iph\ technique and up to 8.54$\times$ compared to the \dbt\ technique. 
Similarly for high contention, we observed up to 40.56$\times$ and 14.51$\times$ 
speedup with respect to the \iph\ and \dbt\ techniques, respectively. 
The greater performance gap is attributed to the fact that in \iph, latching 
contention on the page is increased that is altogether eliminated in \tp. In \dbt, since 
the active set is smaller, and all updates are concentrated to smaller regions, 
the merging frequency is increased, which proportionally reduces the overall 
throughput due to the constant draining of all active transactions.  Finally, due 
to the smaller active set sizes in the medium- and high-contention workloads, the 
cache misses are also reduced as the cache-hit ratio increases. As a consequence, 
the transaction throughput also increases proportionately. 

\begin{figure}[t]
\centering
\includegraphics[trim=0cm 0cm 0cm 0cm, width=3in]{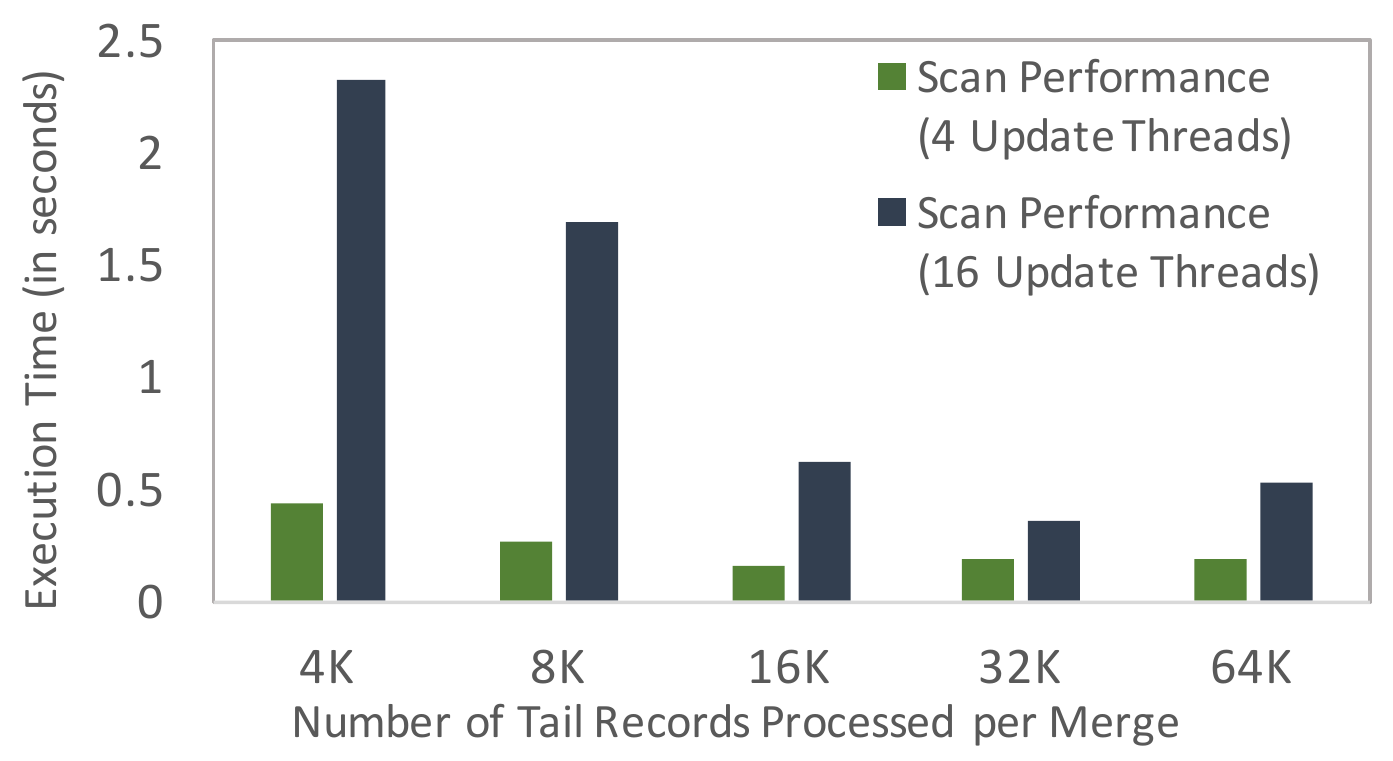}
\caption{Scan performance.}
\label{fig:scan-perf}
\end{figure}

\begin{table}[t]
\centering
 \begin{tabular}{| c | c | c | c |}
  \hline
   & {\tp} & {\small IUH} & {\small DBM} \\ \hline
  {Scan Performance (in secs.)} & 0.24 & 0.28 & 0.38 \\
  \hline
 \end{tabular}
 \caption{Scan performance for different systems.}
\label{table:scan}
\end{table}

\textbf{Scan Scalability:} Scan performance is an important metric for real-time 
OLTP and OLAP systems because it is the basic building block for assembling complex 
ad-hoc queries. We measure the scan performance of \tp\ by computing the SUM aggregation 
on a column that is continuously been updated by the concurrent update transactions.
Thus, the goal of this experiment is to determine whether the merge can keep up 
with high-throughput OLTP workloads. As such, this scenario captures the worst-case 
scan performance because it may be necessary for the scan thread to search for 
the latest values in the merged page or tail pages when the merge cannot cope with 
the update throughput. For columns which do not get updated, the latest values 
are available in the base page itself, as described before. In this experiment 
(Figure~\ref{fig:scan-perf}), we study the single-threaded scan performance with 
one dedicated merge thread. We vary the number of tail records ($M$) that are 
processed per merge (x-axis) and observe the corresponding scan execution time (y-axis) 
while keeping the range partitioning fixed at $64K$ records. We repeat this experiment 
by fixing the number of update threads to 4 and 16, respectively. In general, 
we observe that as we increase $M$, the scan execution time decreases. 
The main reasoning behind this observation is that the scan thread visits tail 
pages for the latest values less often because the 
merge is able to keep up. However, for the smaller values of $M$, the merge is 
triggered more frequently and cannot be sustained. Additionally, the overall cost of 
the merge is increased because the cost of merge is amortized over fewer tail records 
while still reading the entire range of $64K$ base records. Notably, if we delay the merge 
by accumulating too many tail records, then there is slight deterioration in performance. 
Therefore, it is important to balance the merge frequency vs. the amortization cost 
of the merge for the optimal performance, which based on our evaluation, it is when 
$M$ is set to around 50\% of the range size.

\begin{figure}[t]
\begin{center}
\subfigure[Low Contention]{\centering \includegraphics[trim=0cm 0cm 0cm 0cm, width=3in]{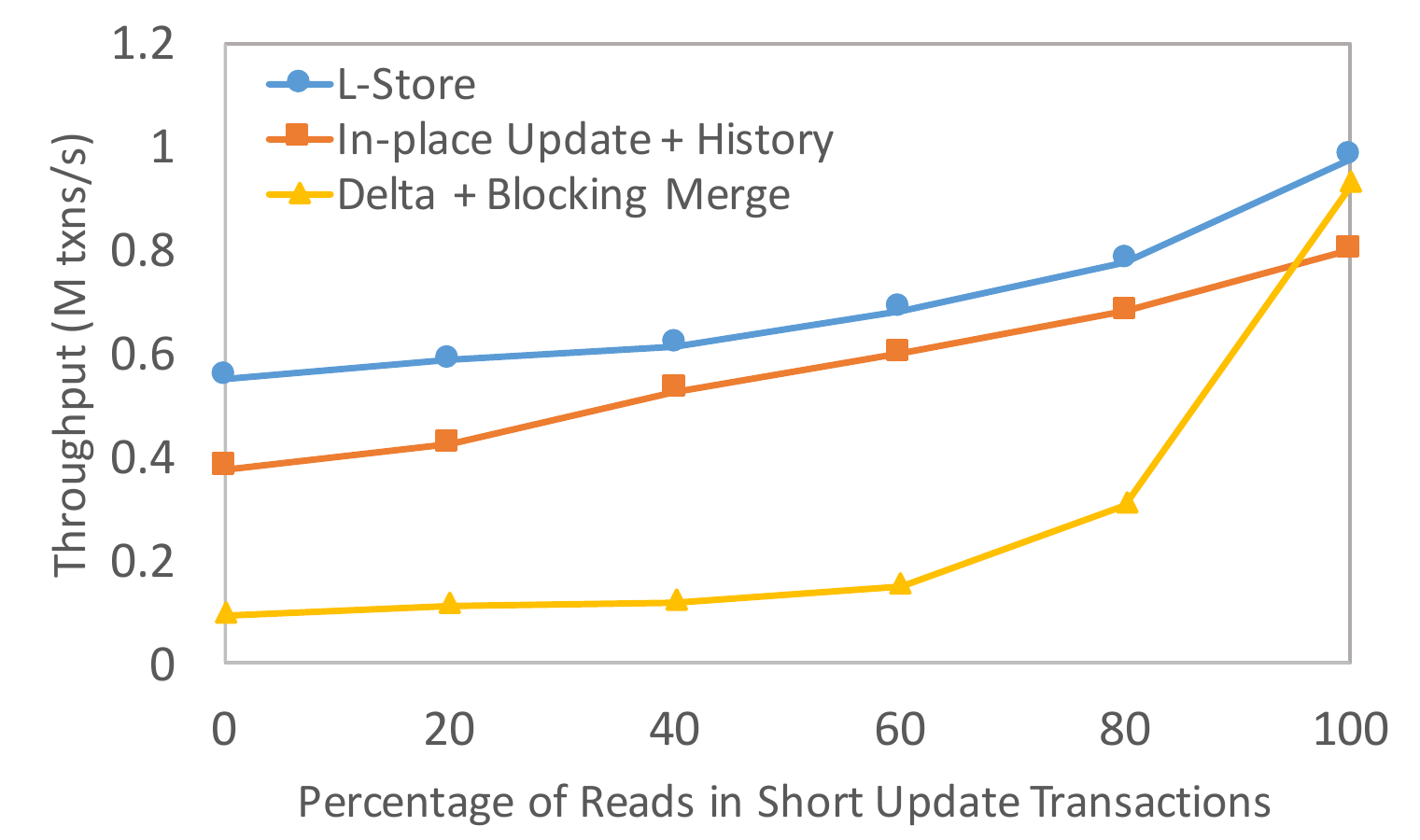} \label{fig:rw-lc}}
\subfigure[Medium Contention]{\centering \includegraphics[trim=0cm 0cm 0cm 0cm, width=3in]{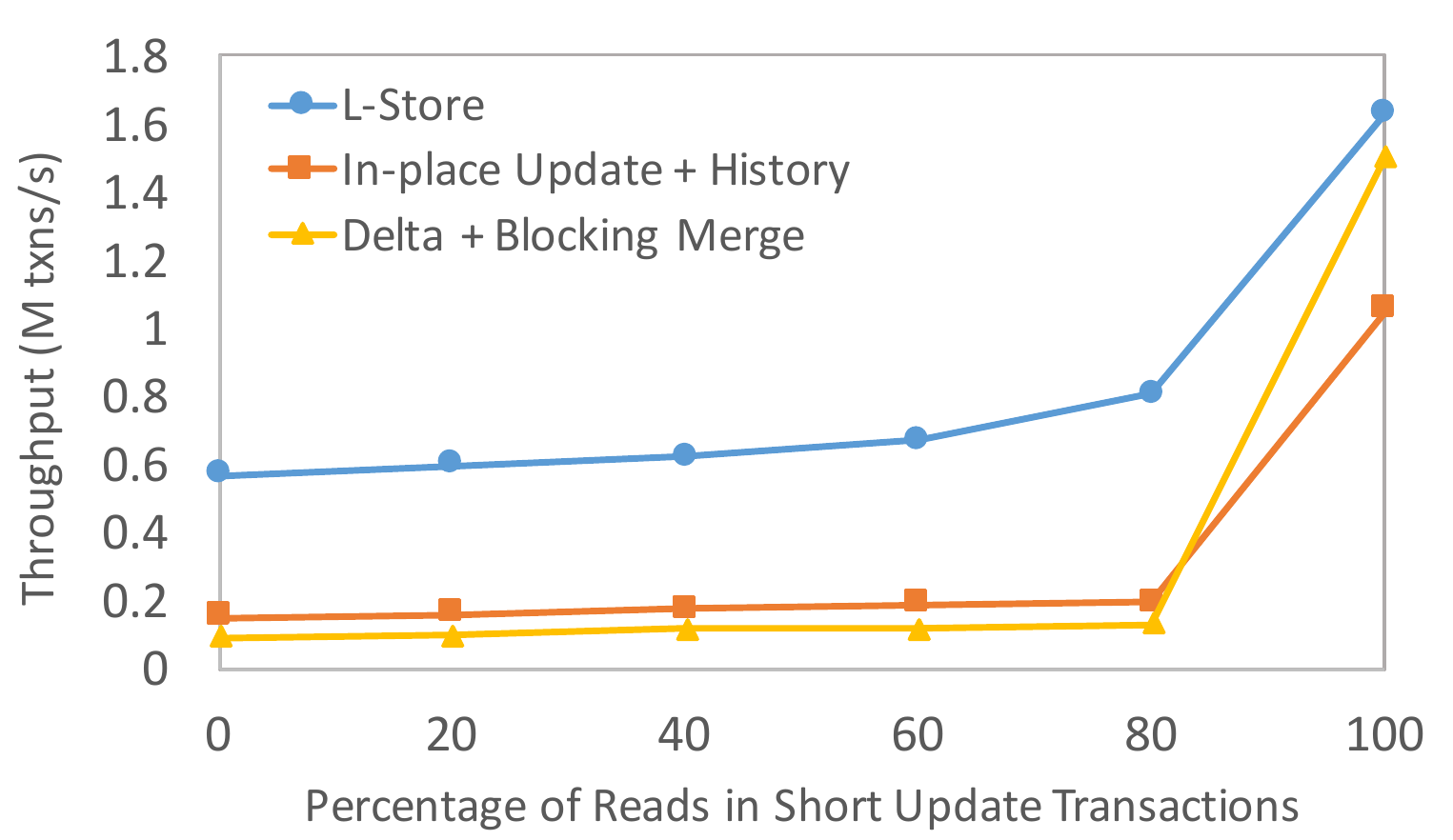} \label{fig:rw-mc}}
\end{center}
\caption{Impact of varying the read/write ratio of short update transactions.}
\label{fig:rw-lc-mc}
\end{figure}

\begin{figure*}[t]
\begin{center}
\subfigure[Update throughput with long-read transactions (Low Cont.).]{\centering \includegraphics[trim=0cm 0cm 0cm 0cm, width=3.15in]{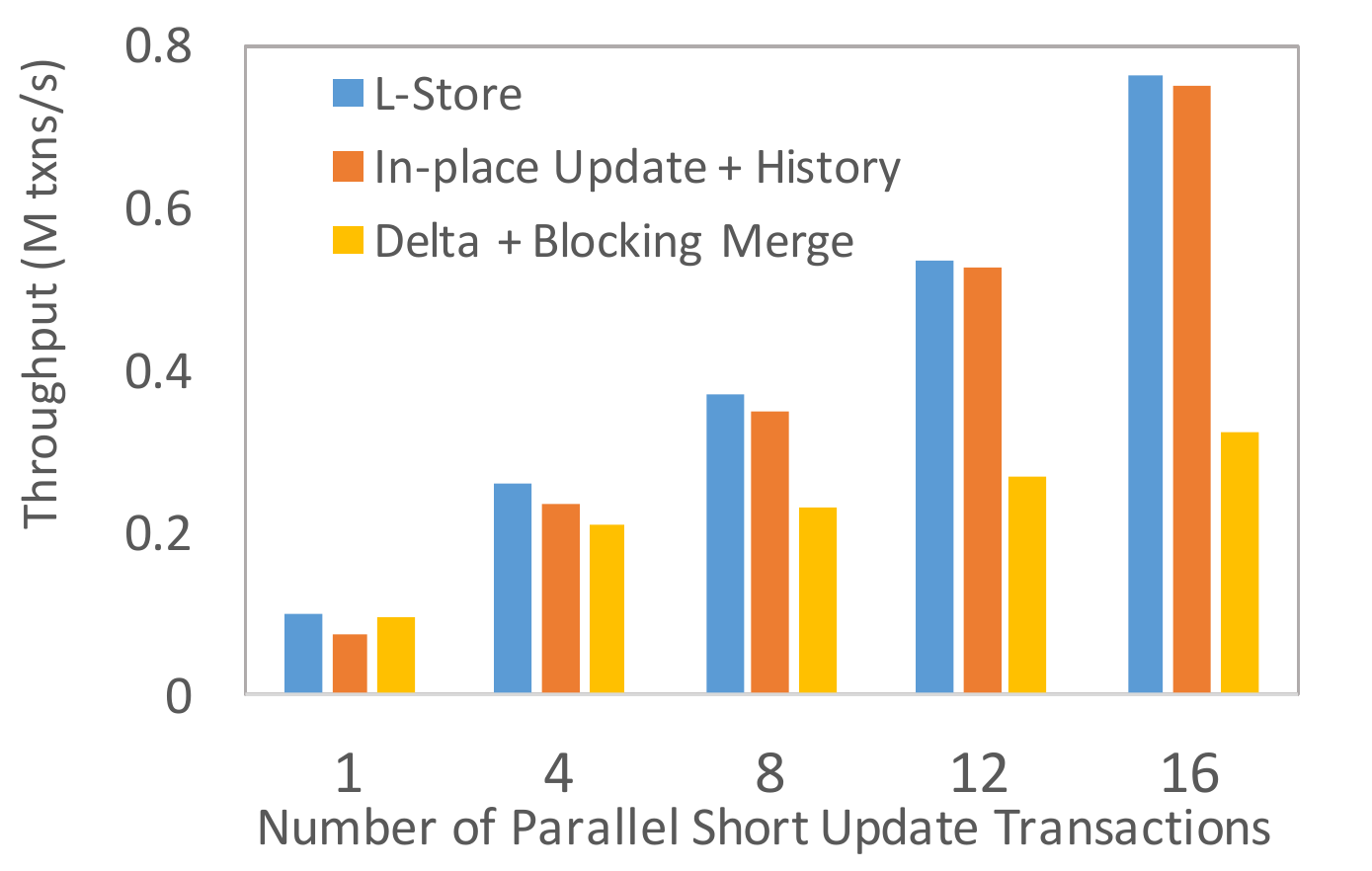} \label{fig:lrt-update-lc}}
\subfigure[Read throughput with short update transactions (Low Cont.).]{\centering \includegraphics[trim=0cm 0cm 0cm 0cm, width=3.15in]{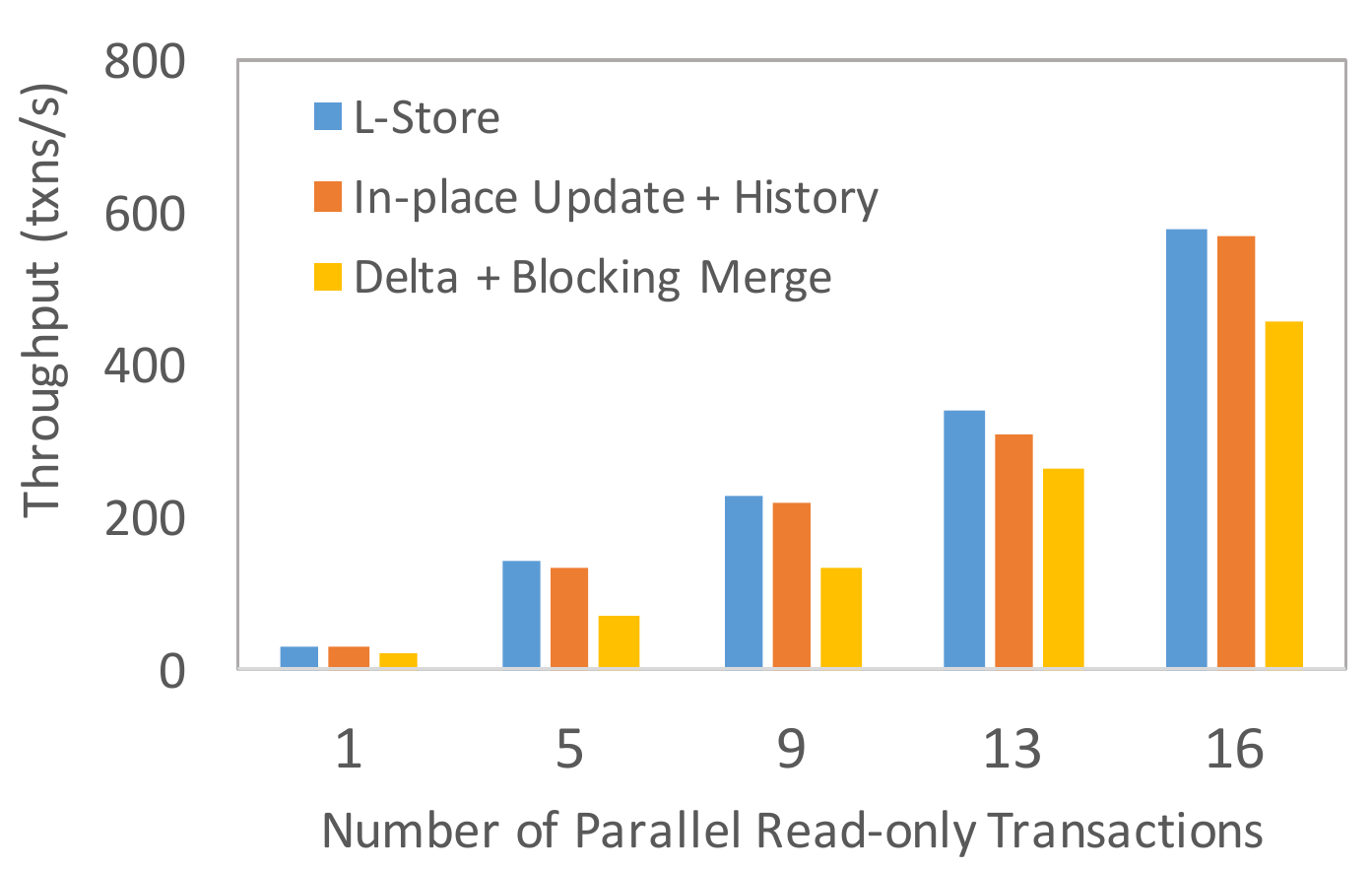} \label{fig:lrt-read-lc}}
\subfigure[Update throughput with long-read transactions (Med Cont.).]{\centering \includegraphics[trim=0cm 0cm 0cm 0cm, width=3.15in]{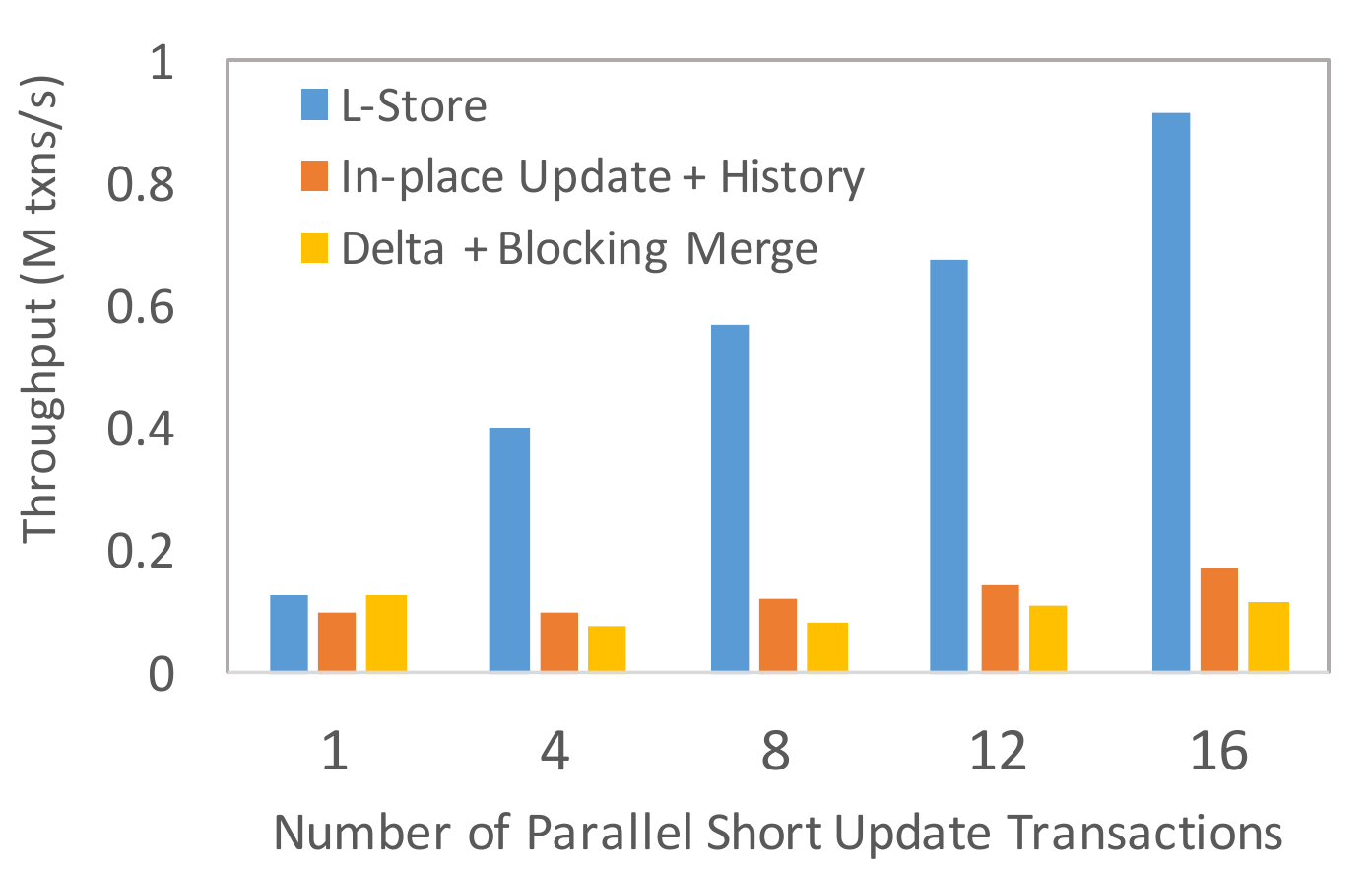} \label{fig:lrt-update-mc}}
\subfigure[Read throughput with short update transactions (Med Cont.).]{\centering \includegraphics[trim=0cm 0cm 0cm 0cm, width=3.15in]{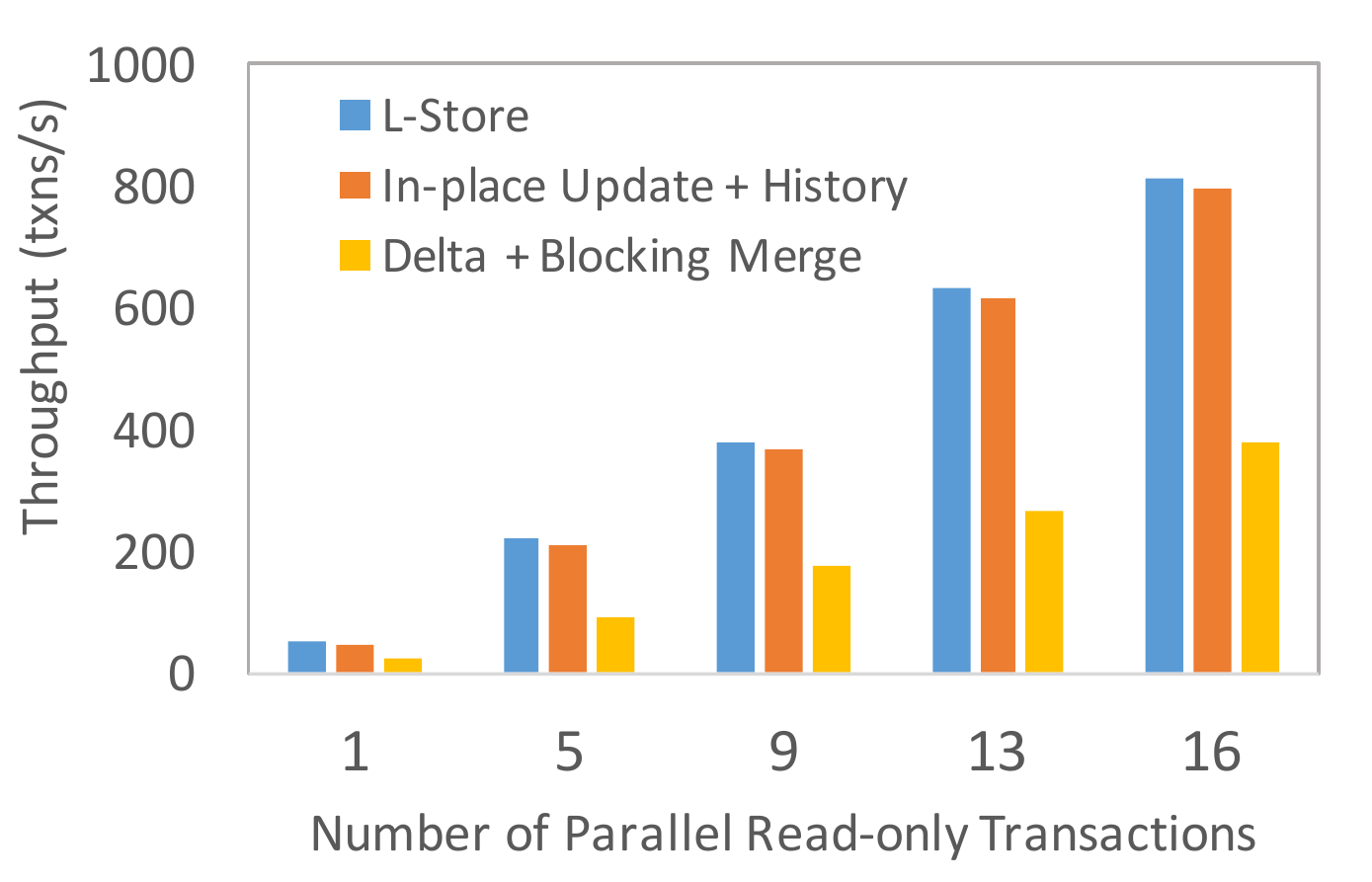} \label{fig:lrt-read-mc}}
\end{center}
\caption{Impact of varying the number of short update vs. long read-only transactions.}
\end{figure*}

We also compare the single-threaded scan performance (for low contention and $4K$ range 
size) of \tp\ with the other two techniques in the presence of 16 concurrent update 
threads (as shown in Table~\ref{table:scan}). Our technique outperforms the \iph\ and \dbt\ techniques 
by 14.28$\%$ and 36.84$\%$, respectively. It is important to note that smaller update range 
sizes, namely, assigning separate tail pages for each $4K$ base records 
instead of $64K$ base records, increases the overall scan performance by improving 
the locality of access within tail pages. Therefore, as mentioned previously in 
Section~\ref{sec_merge_algorithm} and further elaborated in Section~\ref{sec_range}, it 
is beneficial to apply (virtual) fine-grained partitioning over base records (e.g., 
$4K$ records) to handle updates in order to improve locality of access within tail 
pages while applying (virtual) coarser-grained partitioning (e.g., $64K$ records) 
when performing the merge in order to reduce the space fragmentation in the resulting 
merged pages.

\textbf{Impact of varying the workload read/write ratio:}
Short update transactions update only a few records in the database while 
performing reads for a majority of the time. A typical transactional 
workload comprises of $80\%$ read statements and $20\%$ writes~\cite{LarsonBDFPZ11}. 
However, our goal is to explore the entire spectrum from a read-intensive 
workload (read/write ratio 10:0) to a write-intensive workload (read/write ratio 0:10) 
while fixing the number of update threads to 16~\cite{DBLP:journals/pvldb/SadoghiCBNR14}. 
Figure~\ref{fig:rw-lc} shows transaction throughput (y-axis) as the ratio of read-only 
transactions varies in the workload (x-axis) with low contention. As expected, the performance 
of all the schemes increases as we increase the ratio of reads in the transactions because 
contention is a function of writes. As we have more writes in the workload, \iph\ 
technique suffers from increased contention as acquiring read latches conflict 
with the exclusive latches resulting in an extended wait time. The performance 
of the \dbt\ technique also exacerbates since increasing the number of writes increases  
the number of merges performed. This brings down the performance further due to frequent 
halt of the system while draining active transactions. However, note that the gap between 
all of the schemes is the least when the workload consists of 100$\%$ reads. 
In summary, the speedup obtained with respect to \iph\ is up to 1.45$\times$ and up to 5.78$\times$ 
with respect to \dbt\ technique. Note, even for 100$\%$ read, \iph\ continues to pay the 
cost of acquiring read latches on each page.

We repeat the same experiment but restrict the database active set size to 100K rows 
(Figure~\ref{fig:rw-mc}). \tp\ significantly outperforms the other techniques across 
all workloads while varying the read/write ratio. But the performance gap is similar 
with respect to the low contention scenario when there are no update statements 
in the workload. The speedup obtained compared to \iph\ and \dbt\ 
techniques is up to 4.19$\times$ and up to 6.34$\times$ respectively.

\begin{table*}[t]
\centering
 \begin{tabular}{| c | c | c |}
  \hline
   & {\scriptsize \sf L-Store (Column)} & {\scriptsize \sf L-Store (Row)} \\ \hline
  {Scan Performance without updates (in secs.)} & 0.043 & 0.196 \\ \hline
  {Scan Performance with updates (in secs.)} & 0.24 & 0.66 \\
  \hline
 \end{tabular}
 \caption{Scan performance based on row vs. columnar layouts.}
\label{table:scan-rc}
\end{table*}

\begin{table*}[t]
\centering
 \begin{tabular}{| c | c | c | c | c | c |}
  \hline
 
   & { 10\% of Columns} & { 20\% of Columns} & { 40\% of Columns} & { 80\% of Columns} & {All Columns} \\ \hline
   { \scriptsize \sf L-Store (Column)} & 1.46 & 1.35 & 1.17 & 1.08 & 0.98  \\ \hline
   { \scriptsize \sf L-Store (Row)} & 1.45 & 1.45 & 1.45 & 1.45 & 1.45 \\
  
  \hline
 \end{tabular}
 \caption{Point query performance vs. percentage of columns read (M txns/second).}
\label{table:point-query}
\end{table*}

\textbf{Impact of long-read transactions:}
As mentioned previously, it is not uncommon to have long-running read-only transactions in real-time OLTP and OLAP 
systems. These analytical queries touch a substantial part of the database compared to 
the short update transactions, and the main goal is to reduce the interference between 
OLTP and OLAP workloads. In this experiment, we investigate the performance of the different 
schemes in the presence of these long-running read-only transactions, which on an average 
touch 10$\%$ of the base table. We fix the number of concurrent active transactions to 
17 while increasing the number of concurrent read-only transactions from 1 to 16 (the 
short transactions simultaneously vary from 16 to 1). We also allocated a single merge 
thread for \tp\ and \dbt. Figures~\ref{fig:lrt-update-lc}-\ref{fig:lrt-read-lc} represent the scenario for a low 
contention workload, while Figures~\ref{fig:lrt-update-mc}-\ref{fig:lrt-read-mc} represent the scenario for medium 
contention. 

We observe that for both low and medium contention, there is an increase 
in throughput for both long-read transactions and short update transactions when the 
number of threads are increased. Moreover, the performance of read-only transaction  
increases for the medium contention scenario for all the techniques as the updates are 
restricted to a small portion of the database resulting in a higher read throughput. 
In other words, majority of the read-only transactions touch portions of the 
database in which updates do not take place resulting in higher throughput. For read-only 
transactions, our technique outperforms \dbt\ up to 1.97$\times$ and 2.37$\times$ for 
low and medium contention workloads, respectively. For short update transactions, we 
outperform \iph\ and \dbt\ by at most 5.37$\times$ and 7.91$\times$, respectively, 
for medium-content\-ion workload. In the earlier experiments, we had dem\-onstrated that 
\tp\ outperforms other leading approaches for update-intensive workloads, 
and in this experiment, we further streng\-then our claim that \tp\ substantially 
outperforms the leading approaches in the mixed OLTP and OLAP workload as well, 
the latter is due to our novel contention-free merging that does not interfere 
with the OLTP portion of the workload.

\textbf{Impacts of comparing row vs. columnar layouts}
For completeness, we revisit the scan and point query performance while 
considering both row and columnar storage layouts. To enable this comparison, 
we additionally developed a variation of our \tp\ prototype using row-wise 
storage layout, which we refer to as {\small \sf L-Store (Row)}.\footnote{Notably 
our proposed lineage-based storage architecture is not limited to any particular 
data layout; in fact, our technique can be employed even for non-relational data such 
as document or graph data.}

In particular, we compared the single-threaded scan performance (for low contention 
and $4K$ range size) of \tp\ using both row and columnar layouts in the presence of when there 
is no updates or when there are 16 concurrent update threads (as shown in Table~\ref{table:scan-rc}). 
As expected, the scan performance of {\small \sf L-Store (Column)} is substantially higher than 
{\small \sf L-Store (Row)} by a factor of $2.75\times$ and $4.56\times$, with and without updates, 
respectively. Also note that we did not enabled column compression for {\small \sf L-Store (Column)}, 
otherwise even a higher performance gap would be observed because in column stores, an average of 10$\times$ 
compression is commonly expected~\cite{StonebrakerCStore,DBLP:conf/cidr/BonczZN05}. 

We further conducted an experiment with only point que\-ries (on a table 
with 10 columns), where each transaction now consists of 10 read statements, 
and each read statement may read 10\% to 100\% of all columns (as shown in 
Table~\ref{table:point-query}). As expected, the performance of any column 
store is deteriorated as more columns are fetched. When reading only 10-20\% 
of columns, {\small \sf L-Store (Column)} exhibit a comparable throughout as {\small \sf L-Store (Row)}; 
however, as we increase the number of fetched columns, the throughput is decreased. But, 
even in the worst case when all columns are fetched, the throughout only drops by 33\%. 
However, the prevalent observation is that rarely all columns are read or updated 
in either OLTP or OLAP workloads~\cite{pax,StonebrakerCStore,DBLP:conf/cidr/BonczZN05}; 
thus, given the substantial performance benefit of columnar layout for predominant 
workloads, then it is justified to expect a slight throughput decrease in  
rare cases of when point queries are forced to access all columns.

\section{Related Work}
In recent years, we have witnessed the development of many in-memory engines 
optimized for OLTP workloads either as research prototypes such as HyPer~\cite{Kemper:2011,hyperMVCC} 
and ES2~\cite{Cao:2011} or for commercial use such as Microsoft Hekaton~\cite{Diaconu:2013}, 
Oracle Database In-Memory~\cite{LahiriOracle}, VoltDB~\cite{Stonebraker_thevoltdb}, and HANA~\cite{Krueger11,Plattner14}. Most 
of these systems are designed to keep the data in row format and in the main memory to increase 
the OLTP performance. In contrast, to optimize the OLAP workloads, columnar format is 
preferred. The early examples of these engines are C-Store~\cite{StonebrakerCStore} 
and MonetDB~\cite{DBLP:conf/cidr/BonczZN05}. Recently, major database vendors also started 
integrating columnar storage format into their existing database engines. SAP 
HANA~\cite{DBLP:journals/debu/FarberMLGMRD12} is designed to handle both OLTP 
and OLAP workloads by supporting in-memory columnar format. IBM DB2 BLU~\cite{Raman:2013} 
introduces a novel compres\-sed columnar storage that is mem\-ory-optimized (and not 
restricted to being only memory-resident) that substantially improves the execution of 
complex analytical workloads by operating directly on compress\-ed data. 
Below, we summarize the key aspects of recent developments that also aim to 
support real-time OLTP and OLAP workloads on the same platform.

HyPer, a main-memory database system, guarantees the ACID properties of OLTP 
transactions and supports running OLAP queries on consistent snapshot~\cite{Kemper:2011}. 
The design of HyPer leverages OS\--proc\-essor-cont\-roll\-ed lazy copy-on-write 
mechanism enabling to create a consistent virtual memory snapshot. Unlike \tp, HyPer 
is forced to running transactions serially when the workload is not partitionable. Notably, 
HyPer recently employed multi-version concurrency to close this gap~\cite{hyperMVCC}.  
El\-astic power-aware data-intensive cloud computing platform (epiC) was designed to provide 
scalable database services on cloud \cite{Cao:2011}. epiC is designed to handle both 
OLTP and OLAP workloads~\cite{chen2010providing}. However, unlike \tp, the OLTP 
queries in ES2 are limited to basic get, put, and delete requests (without 
multi-statements transactional support). Furthermore, in ES2, it is possible that 
snapshot consistency is violated and the user is notified subsequently~\cite{Cao:2011}.

Microsoft SQL Server currently contains three types of engines: the classical SQL 
Server engine designed to process disk-based tables in row format, the Apollo engine 
designed to maintain the data in columnar format~\cite{Larson:2011:SSC:1989323.1989448}, 
and the Hekaton in-memory engine designed to improve OLTP workload performance~\cite{Diaconu:2013,Larson:2015}. 
In contrary, our philosophy is to avoid cost of maintaining multiple engines and to 
introduce a unified architecture to realize real-time OLTP and OLAP capabilities. 
Noteworthy, Microsoft has also recently announced moving towards supporting 
real-time OLTP and OLAP capabilities starting 2016~\cite{Larson:2015}. However, 
to provide OLTP and OLAP support among loosely integrated engines, a rather 
complex procedure is required as expected~\cite{Larson:2015}. In particular, data 
are forced to move between Hekaton (a row-based engine) and the columnar 
indexes (based on the Apollo engine) in a cyclic fashion by creating a number 
of foreground transactions (e.g., one large transaction touching a million 
records on average followed by thousands of smaller transactions) that ultimately 
ran concurrently and creates potential contention with the users' transactions~\cite{Larson:2015}.
In contrast, in \tp, we rely on purely columnar storage (while no multiple copies 
of the data is maintained) and, more importantly, our consolidation is based on a 
novel con\-tention-free merge process that is performed asynchronously and completely 
in the background, and the only foreground task is pointer swaps in the page directory 
to point to the newly created merged pages. 

Oracle has recently announced a new product called Ora\-cle Data\-base In\--Mem\-ory that 
offers dual-format option to support real-time OLTP and OLAP. In the Oracle 
architecture, data resides in both columnar and row formats~\cite{LahiriOracle}. 
To avoid maintaining two identical copies of data in both columnar and row format, 
a ``layout transparency'' abstraction was introduced that maps database into 
a set of disjoint tiles (driven by the query workload and the age of data), where 
a tile could be stored in either columnar or row format~\cite{Arulraj:2016:BAR:2882903.2915231}. 
The key advantage of the layout-transparent mapping is that the query execution runtime 
operates on the abstract representation (layout independent) without the need to 
create two different sets of operators for processing the column- and row-oriented data. 
In the same spirit, SnappyData proposed a unified engine to combine 
streaming, transaction, and analytical processing, but from the storage perspective, 
it maintains recent transactional data in row format while it ages data to a columnar 
format for analytical processing~\cite{Ramnarayan:2016:SHS:2933267.2933295}. SnappyData 
employed data ageing strategies similar to the original version of SAP 
HANA~\cite{Sikka:2012:ETP:2213836.2213946}. 

Contrary to the aforementioned efforts, in \tp, we strictly keep only one copy and one 
representation of data; thus, fundamentally eliminating the need to maintain 
layout-independent mapping abstraction and storing data in both columnar and row formats.   
Last but not least, (academically-led) HANA~\cite{Krueger11,Plattner14} also strives 
to achieve real-time OLTP and OLAP engine, most notably, we share the same philosophy 
governing HANA that aims to develop a generalized solution for unifying OLTP and 
OLAP as opposed to building specialized engines. However, what distinguishes our 
architecture from HANA is that we propose a holistic columnar storage without 
the need to distinguishing between a main store and a delta store. In addition, 
we propose a contention-free merge process, whereas in HANA, the merge process is 
forced to drain all active transactions at the beginning and end of the merge 
process~\cite{Krueger11}, a contention that results in a noticeable slow down 
as demonstrated in our evaluation.

On a different front, database concurrency theory, an old age problem, has recently 
been revived by industry (e.g., \cite{LarsonBDFPZ11,Diaconu:2013}) and academia (e.g., \cite{DBLP:journals/pvldb/KallmanKNPRZJMSZHA08,DBLP:conf/sigmod/ThomsonDWRSA12,DBLP:journals/pvldb/SadoghiCBNR14,DBLP:conf/icdcs/0001SJ16}) due to hardware tre\-nds 
(e.g., multi-cores and large main memory) and application requirements 
(e.g., the need for processing millions of transactions per second in targeted advertising 
and algorithmic trading~\cite{Sadoghi:2010:EEP:1920841.1921029}). Microsoft Hekaton focuses 
primarily on optimistic concurrency by assuming that roll backs are 
inexpensive and conflicts are rare~\cite{LarsonBDFPZ11}. Hekaton avoids the use of a lock 
manager and relies on read validation to ensure repeatable reads, performs re-execution 
of all range queries to achieve serializability, and detects write-write conflicts by using 
CAS operator and aborting the second writer to avoid any blocking. Furthermore, going 
beyond Hekaton, a 2-version concurrency control (2VCC) was proposed in~\cite{DBLP:journals/pvldb/SadoghiCBNR14} 
that offers efficient (latch-free) pessimistic and optimistic models that co-exists peacefully. 
Similar to the approach in Hekaton~\cite{LarsonBDFPZ11}, the 2VCC model~\cite{DBLP:journals/pvldb/SadoghiCBNR14} also rejects the idea 
of tuning the concurrency model to only limited types of workloads such as partitionable workloads 
(a direction that is pursed by ~\cite{DBLP:journals/pvldb/KallmanKNPRZJMSZHA08,DBLP:conf/sigmod/ThomsonDWRSA12}). Since \tp's focus is to provide a general 
storage architecture, any concurrency models can be employed; in particular, we relied 
on the optimistic concurrency model proposed in~\cite{DBLP:journals/pvldb/SadoghiCBNR14} 
while supporting the speculative reads proposed in~\cite{LarsonBDFPZ11}.

\section{Conclusions}
In this paper, we develop \tpfull\ to realize real-time OLTP and OLAP processing 
within a single unified engine. The key features of \tp\ can succinctly be 
summarized as follows. In \tp, recent updates for a range of records are strictly 
appended and clustered in its corresponding tail pages to eliminate read/write 
contention, which essentially transforms co\-stly point updates into an amortized,  
fast analytical-like update query. Furthermore, \tp\ achieves (at most) 2-hop 
access to the latest version of any record through an effective embedded indirection 
layer. More importantly, we introduce a novel contention-free and relaxed merging of only 
stable data in order to lazily and independently bring base pages (almost) up-to-date without 
blocking on-going and new transactions. Furthermore, every base page relies on independently 
tracking the lineage information in order to eliminate all coordination and recovery 
even when merging different columns of the same record independently. Lastly, a novel contention-free 
page de-allocation using epoch-based approach is introduced without interfering 
with ongoing transactions. In our evaluation, we demonstrate that \tp\ outperforms 
\iph\ by factor of up to 5.37$\times$ for short update transactions while achieving 
slightly improved performance for scans. It also outperforms \dbt\ by 7.91$\times$
for short update transactions and up to 2.37$\times$ for long-read analytical queries.

\section{Acknowledgments}
We wish to thank C. Mohan, K. Ross, V. Raman, R. Barber, R. Sidle, A. Storm, X. Xue, I. Pandis, Y. Chang, and G. M. Lohman 
for many insightful discussions and invaluable feedback in the earlier stages of this work.



\end{document}